\documentclass[a4paper]{article}

	\usepackage[standard,thmmarks]{ntheorem}
	\usepackage{amsmath,xspace,hyperref,cleveref,a4,xcolor}
	\usepackage{todonotes,pifont,caption,subcaption,authblk}
        \usepackage[all]{xy}
	\usepackage[british]{babel}
	\renewtheorem{theorem}{Theorem}[section]
	\renewtheorem{proposition}[theorem]{Proposition}
	\renewtheorem{lemma}[theorem]{Lemma}
	\renewtheorem{corollary}[theorem]{Corollary}
	\theorembodyfont{\upshape}
	\renewtheorem{example}[theorem]{Example}
	\renewtheorem{remark}[theorem]{Remark}
	\renewtheorem{definition}[theorem]{Definition}
	\newcommand{\parspace}{parametrised space}
	\newcommand{\sdone}{\texttt1}
	\newcommand{\sdzero}{\texttt0}
	\newcommand{\XX}{\mathbf X}
	\newcommand{\YY}{\mathbf Y}
	\newcommand{\ZZ}{\mathbf Z}
	\newcommand{\Z}{\mathbb Z}
	\newcommand{\RR}{\mathbb R}
	\newcommand{\lb}{\operatorname{lb}}
	\newcommand{\irram}{\texttt{iRRAM}\xspace}
	\newcommand{\bigo}[1]{\mathcal O \left(#1\right)}
	\newcommand{\demph}{\textbf}
	\newcommand{\str}{\mathbf}
	\newcommand{\abs}[1]{\left| #1\right|}
	\newcommand{\length}{\abs}
	\newcommand{\B}{\mathcal B}
	\newcommand{\DD}{\mathbb D}
	\newcommand{\ID}{\mathrm I\mathbb D}
	\newcommand{\NN}{\mathbb N}
	\newcommand{\mto}{\rightrightarrows}
	\newcommand{\sop}{\mathbf{SOP}}
	\newcommand{\Set}[2]{\left\{ #1 \mid #2 \right\}}
	
	\newcommand{\reg}{\mathrm{LM}}
	\newcommand{\KCrep}{Kawamura-Cook representation\xspace}
	\newcommand{\KCspace}{Ka\-wa\-mu\-ra-Cook space}
	\newcommand{\norm}[1]{\left\|#1\right\|}
	\renewcommand{\mod}{\operatorname{mod}}
	\newcommand{\ie}{\textit{i.e.}, }
	\DeclareMathOperator{\dom}{dom}
	\DeclareMathOperator{\Time}{time}
	\DeclareMathOperator{\diam}{diam}
	\DeclareMathOperator{\eval}{eval}

\date{}
\title{Parametrised second-order complexity theory\\ with applications to the study of interval computation\textsuperscript{\small1}}
\author{Eike Neumann}
\affil{Aston University \\
  School of Engineering \& Applied Science\\
  Aston Triangle\\
  Birmingham B4 7ET}
\author{Florian Steinberg\textsuperscript{\small2}}
\affil{INRIA Saclay \\
  Toccata Team\\
  B\^ at 650, Rue Noetzlin\\
  91190 Gif-sur-Yvette
}

\begin{document}
	\maketitle
	\abstract
	{
		We extend the framework for complexity of operators in analysis devised by Kawamura and Cook (2012)
		to allow for the treatment of a wider class of representations.
		The main novelty is to endow represented spaces of interest with an additional function on names,
		called a parameter,	which measures the complexity of a given name.
		This parameter generalises the size function which is usually used in
		second-order complexity theory and therefore also central to the framework of Kawamura and Cook.
		The complexity of an algorithm is measured in terms of its running time as a second-order
		function in the parameter,
		as well as in terms of how much it increases the complexity of a given name,
		as measured by the parameters on the input and output side.

		As an application we develop a rigorous computational complexity theory for interval computation.
		In the framework of Kawamura and Cook the representation of real numbers based on nested interval enclosures
		does not yield a reasonable complexity theory.
		In our new framework this representation	is polytime equivalent to
		the usual Cauchy representation based on dyadic rational approximation.
		By contrast, the representation of continuous real functions based on interval enclosures is
		strictly smaller in the polytime reducibility lattice
		than the usual representation, which encodes a modulus of continuity.
		Furthermore, the function space representation based on interval enclosures is optimal
		in the sense that it contains the minimal amount of information
		amongst those representations which render evaluation polytime computable.
	}
	\footnotetext[1]{\includegraphics[scale=0.04]{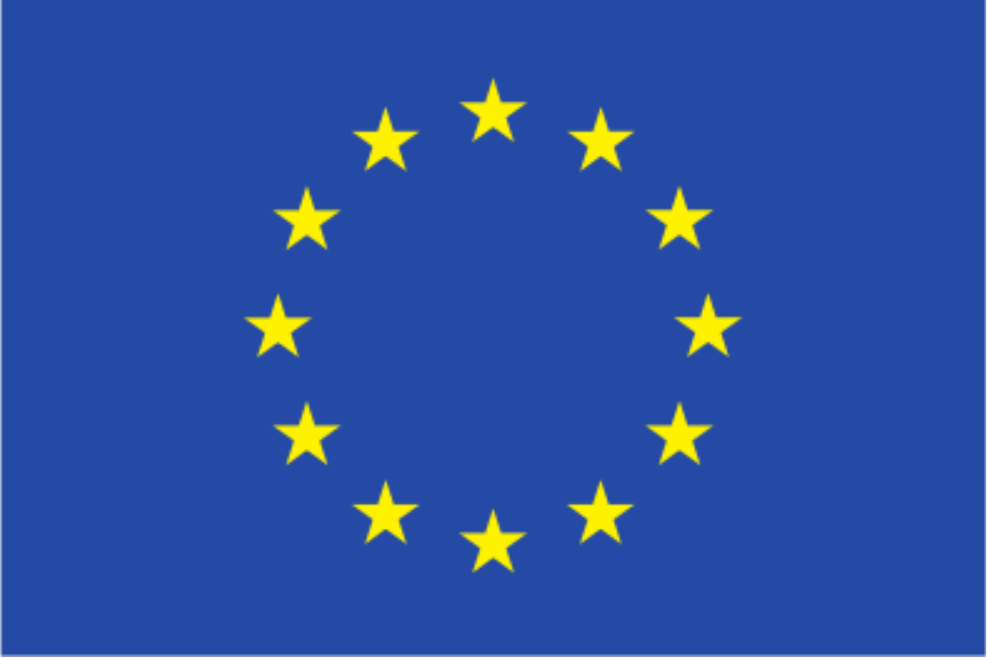} This project has received funding from the European Union’s Horizon 2020 research and innovation programme under the Marie Sk\l{}odowska-Curie grant agreement No 731143.}
	\footnotetext[2]{
	The second author was supported by the ANR project \emph{FastRelax}(ANR-14-CE25-0018-01) of the French National Agency for Research.}
	\tableofcontents
	\newpage
	\section{Introduction}\label{sec:introduction}

	Computable analysis is an extension of the theory of computation over the natural numbers
	to continuous data, such as real numbers and real functions, based on the Turing machine
	model of computation.
	Computability of real numbers is studied already in Turing's paper \cite{turing1936computable} on the halting problem.
	Computability of real functions was introduced by
	Grzegorczyk \cite{Grzegorczyk, Grzegorczyk2, Grzegorczyk3} and Lacombe \cite{MR0072080}.
 	Kreitz and	 Weihrauch \cite{KreitzWeihrauch, MR1795407} introduced a general theory of computation on second-countable $T_0$-spaces.
 	This was further generalised by Schr\"oder \cite{SchroederPhD, MR1923914} to $T_0$ quotients of countably based spaces,
  which constitute in a certain sense the largest class of topological spaces which can be endowed with a reasonable
  computability structure \cite[Theorem 13]{MR1923914}.

	One of the goals of computable analysis is to provide mathematically rigorous semantics for computation over
	continuous data structures.
	Algorithms in numerical analysis are usually described using the real number model,
	where real numbers are regarded as atomic entities.
	A widely used mathematically rigorous formalisation of this idea
	is the Blum-Shub-Smale \cite{BSS, BSS2} machine.
	Such algorithms cannot be implemented directly on a digital computer,
	as real numbers cannot be encoded with a finite number of bits.
	The usual substitution for real numbers are floating point numbers,
	which behave quite differently.
	For instance, addition and multiplication on floating point numbers are not even associative.
	Thus,	the behaviour of floating point algorithms depends on phenomena that are absent in the
	real number model, such as numerical stability	and error propagation.
	These issues have to be studied separately, which usually requires a substantial amount of additional effort.
	Even then, the precise contract that an implementation fulfils is usually not fully specified
	and the semantics of the implementation remain vague.
	As a consequence the semantics of an algorithm can differ considerably from the semantics
	of its implementation and different implementations may well have different semantics.

	By contrast, any algorithm based on computable analysis can be implemented directly on a physical computer.
	It consists of a rigorous specification of input and output, so that it precisely describes the steps
	that have to be taken to obtain the desired result to a given accuracy.
	Software packages based on computable analysis include
	\irram \cite{iRRAM}, Ariadne \cite{Ariadne}, AERN \cite{AERN},
	and RealLib \cite{MR2275414}.

	For the study of practical algorithms it is clear that computational complexity should play a
	central role.
	Whilst the notion of computability over continuous data is robust, well understood,
	and universally accepted in the computable analysis community,
	computational complexity in analysis is far less developed
	and even some basic definitions are the subject of ongoing debate.
	The study of computational complexity in analysis was initiated by Friedman and Ko \cite{KoFriedman}.
	They defined the computational complexity of real numbers and of real functions on compact intervals
	and proved some famous hardness results for problems such as integration, maximisation, or solving
	ordinary differential equations
	(see \cite{FriedmanMaxInt, KoODE}, cf.~also \cite{KawamuraODE}).
	This line of research is summarised nicely in Ko's book \cite{MR1137517}.
	The main gap in the work of Friedman and Ko is that,
	while	their definition of computational complexity for real functions carries over to functions on
	compact metric spaces, it does not generalise easily to functions on non-compact spaces.
	In practice one is most interested in the study of operators on infinite dimensional vector spaces,
	such as spaces of continuous functions, $L^p$-spaces or Sobolev spaces.
	The aforementioned hardness results concern operators of this kind,
	but they are non-uniform in the sense that they establish
	that the operators in question map certain feasibly computable functions to functions of
	conjecturally
	\footnote{assuming standard conjectures in computational complexity such as P $\neq$ NP.}
	high complexity.
  While such non-uniform lower bounds are sufficient to show that an operator is infeasible,
	it remained unclear which operators should be considered feasible.

	One of the main reasons for such a notion not being available was
	the lack of an accepted notion of feasibility for second-order functionals.
	A candidate solution had been proposed by Mehlhorn already in 1975 \cite{MR0411947}:
	The class of \demph{basic feasible functionals}, which he defined
	by means of a generalisation of a limited recursion scheme
	that leads to a  well-known characterisation of the polytime
	computable functions on the natural numbers.
	However, it remained a point of debate for a long time
	to which extent this class fully captures the intuitive notion of
	feasibility \cite{Cook1992}.
	Further investigations into this topic revealed the type-two basic feasible functionals
	to be a very stable class that became established as the foundation of second-order
	complexity theory \cite{Pezzoli1998,irwin_royer_kapron_2001}.
	An important step in this process, and something that opened the field up for
	applications, was the characterization of the basic feasible functionals by
	means of resource bounded oracle Turing machines due to Kapron and Cook \cite{MR1374053}.
	Based on this characterization,
	Kawamura and Cook introduced a framework for complexity of operators in analysis \cite{Kawamura:2012:CTO:2189778.2189780}
	that generalises the definition of feasibly computable functions
	of Friedman and Ko to a wider class of spaces, including the aforementioned examples.
	This kicked off a series of investigations \cite[and many more]{MR3239272,postive,lmcs:3924,10.1109/LICS.2017.8005139}.

	However, there remains a gap between theory and practice.
	Within the framework of Kawamura and Cook it is impossible to model the behaviour of software based on
	computable analysis such as the libraries mentioned above.
	The reason for this is that all these implementations are based on interval arithmetic
	or extensions thereof, such as Taylor models.
	The representations which underlie these approaches are known to exhibit highly pathological behaviour
	within the framework of Kawamura and Cook \cite{MR3219039,SteinbergPhD}
	and
	those representations which can be used within their framework do not always seem to be an appropriate
	substitute.
	For instance, in the Kawamura-Cook model any representation of the space of continuous functions
	which renders evaluation polytime computable
	also allows for the computation of some modulus of continuity of a given function in polynomial time.
	In \irram this requires an exponential exhaustive search, see \cite{Aminimal}.

	The present work is an attempt to bridge this gap by extending the framework of Kawamura and Cook
	in order to develop a meaningful complexity theory for a broader class of representations.

	We do so by endowing a represented space $(X, \xi)$ with an additional function
	$\mu \colon \dom (\xi) \to \NN^{\NN}$, called the \demph{parameter},
	which is intended to measure the complexity of the names of elements of $X$.
	These parameters are a generalisation of the size function which is used to measure complexity
	of string functions in second-order complexity theory.
	The pair $(\xi, \mu)$ is called a \demph{parametrised representation}
	and the triple $(X, \xi, \mu)$ is called a \demph{preparametrised space}.
	The complexity of an algorithm for computing a function $f \colon X \to Y$ between
	preparametrised spaces is measured by the
	dependence of the running time of the algorithm in terms of the parameter of the input name
	and by the growth in size of the parameter of the output name compared with the parameter of the input name.
	As in the Kawamura-Cook model, polytime computability is defined using second-order polynomials:
	An algorithm runs in polynomial time if and only if both its running time and the parameter of the output name
  are bounded by a second-order polynomial in the parameter of the input name.
	A preparametrised space is called a \demph{\parspace} if its identity function is polynomial
	time computable.

	The pathological behaviour of the representation of real numbers based
	on interval enclosures is eliminated by the natural choice of a parameter for this representation.
	The resulting parametrised representation is polytime equivalent to the usual Cauchy representation
	with the size function as parameter	(\Cref{resu:equivalence to the Cauchy representation}).
	In particular, a real function is polytime computable with respect to the parametrised interval 
	representation if and only if it is polytime computable in the usual sense.
	While this result might suggest that nothing much is gained from this new definition,
	we show that the natural uniform complexity structure on the space of real functions
	viewed as a parametrised space is different from the complexity structure induced by the Kawamura-Cook 
	representation,
	and that the complexity induced by the natural parametrised space structure 
	corresponds more closely to the complexity of operators in practical implementations.
	We define a parametrised space $C(I)_i$ of continuous functions based on interval
	enclosures, and show that this space most precisely reflects the behaviour of functions
	expected from implementations:
	On one hand, it leads to the right notion of polytime computability of functions
	(\Cref{resu: polytime computability of functions}), function evaluation is polytime
	computable (\Cref{resu:evaluation}) and the same is	true for many of the usual operations
	one wants to compute quickly (\Cref{resu:operations}, \Cref{resu: more operations}).
	On the other hand, finding a modulus of continuity, which is notoriously slow in
	practice, is provably not polytime computable (\Cref{resu:modulus}).
        It is proved in the appendix that this space is isomorphic to a natural model of the space of
        real functions used in \irram (\Cref{resu: irram funs and interval funs}).

	We investigate the \parspace\ of interval functions further and prove
	that any other parametrised representation of this space
	such that evaluation is polytime computable can be translated to it (\Cref{resu:minimality}).
	There are two reasons why we consider this result to be especially important:
	Firstly it resembles a result that Kawamura and Cook proved about a representation they introduced,
	which is currently considered to be the standard representation for the continuous
	functions on the unit interval for this reason.
	We compare their representation to the representation using interval enclosures and show that it sits
	strictly higher	in the lattice of polytime translatability (\Cref{resu: interval function < Kawamura-Cook function}).
	This reflects the fact that the minimality result by Kawamura and Cook ranges over a restricted class
	of \parspace s.
	We characterise the spaces they consider as essentially those
	that have a polytime computable parameter (\Cref{resu: Characterisation of Kawamura-Cook space}).
	The second reason why we consider the minimality result for interval functions
	to be important is that it includes a quantification over all parametrised representations.
	This demonstrates that the definition of a \parspace\ is not chosen too general to allow for meaningful results.
	Throughout the paper we provide more support for our belief that \parspace s are a good general framework
	for complexity considerations in computable analysis
	(for instance \Cref{resu: polytime function on interval reals has first-order time bound}).

\paragraph{Related work.}

	Parameters and parametrised complexity in our sense are present
	in the work of Rettinger \cite{RettingerCCA} and Lambov \cite{LambovParameter}.
	Rettinger works in a different setting, avoiding second-order polynomials and we
	significantly add to and modify Rettinger's ideas.
	Lambov's work includes a good part of our results on interval representations in a different language.
	However, Lambov does not attempt to build a general framework of parametrised complexity in analysis,
	which requires further ideas that are not present in his work.
	We hence believe that the present work extends his results considerably.

	For a restricted case some of the core definitions proposed in this paper are also present
	in the work of Kawamura, M\"uller, R\"osnick and Ziegler \cite{arxiv1211.4974}.
	The authors introduce parameter functions with integer values.
	This covers only a very special case of pre\parspace s,
	namely those where the value of the parameter on a name is a constant function.
	More significantly, their applications all make the value of the parameter accessible to the algorithm
	and therefore allow	for formulation in the framework of Kawamura and Cook by modification of the representations involved.
	For their applications this is unavoidable,
	as the functions they consider fail to be feasibly computable
	unless the parameter is provided as extra advice.
	Most of the content of the preprint \cite{arxiv1211.4974} was published in \cite{KAWAMURA2015689}.
	Unfortunately, in the published version the authors decided to
	further restrict the definition of parameters, by making it a
	function on the represented space, rather than a function
	on the domain of the representation.
	They have to recover the desired behaviour by first introducing suitable
	covering spaces of the spaces of interest.
	This makes it a lot more difficult to see the connection between their work and ours.
	Similar applications of parametrised complexity are found
    for instance in \cite{Thies1, Thies2, Pouly}.

	We also consider work by Schr\"oder to be related \cite{MR2090390}.
	There are two ways in which his results relate to the contents of this paper:
	The first connection is that he equips a representation with an additional
	integer-valued size function that can, just like in the previous paragraph,
	be considered a special case of a parameter.
	The second connection is more interesting but also more difficult to make:
	Schr\"oder provides conditions for represented spaces
	under which every machine that computes a function
	between these spaces has a well-defined first-order running time.
	This can be interpreted as devising a pair of a representation and a parameter for
	the computable functions between two spaces such that evaluation is polytime computable:
	The representation takes an index of a machine computing a realiser of a function to
	that function and the parameter assigns to such an index the running time of the machine it encodes.
	The running time of the evaluation operation is the overhead needed for simulating a machine.
	However, one should note that these observations are not explicitly stated in \cite{MR2090390}.
	Also, Schr\"oder does not consider polytime computability but only reasons
	about the existence of time bounds in general.

	Kawamura and Pauly \cite{MR3219039} study exponential objects in the category
	of polytime computable mappings.
	They consider one of the standard function space constructions from computable
	analysis, which is obtained by encoding a continuous function by an index
	of a Turing machine together with an oracle, such that the machine computes
	the function relative to the oracle.
	They show that well-behaved function spaces can be constructed for a class
	of spaces which they call ``effectively polynomially bounded''.
	This class of spaces is very similar to the one considered by Schr\"oder.
	Their function space construction can be viewed as an extension of the
	construction sketched in the previous paragraph by adding arbitrary oracles.
	Like Schr\"oder they measure the size of objects in effectively polynomially
	bounded spaces by means of a parameter function with values in the natural
	numbers.
	Their work is also the only example in the literature that we are aware of
	that discusses the issue of polytime computability of the identity function
	on a \parspace.
	Curiously, in the published version of the paper
	the connections to parametrised complexity are significantly obscured.
	The connection to our work is much more visible in an early preprint
	\cite{arxiv1401.2861v1}.

\paragraph{Outline of the paper.}

	Section \ref{sec:introduction} recalls the most important concepts from second-order complexity theory.
	The approach via resource bounded oracle Turing machines is chosen, and this is essential for rest of the paper.
	The discussion of the framework of Kawamura and Cook
	, which imposes an additional condition on the names,
	is postponed to the later Section \ref{sec:comparison to represented spaces}.
	The second part \ref{sec:notations and complexity on the reals} recalls the
	Cauchy and the interval representation of the real numbers and discusses
	how second-order complexity theory can be applied
	to find a well behaved notion of complexity of the former
	and why the same approach fails for the latter.

	Section \ref{sec:parametrised spaces} introduces the general concept of \parspace.
	Polytime computable functions are introduced using second-order polynomials.
	It is shown that they are closed under composition.
	Section \ref{sec: interval reals} applies this to the representation of real numbers based on sequences of nested intervals.
	It is shown that the \parspace\ of interval reals is polytime isomorphic to
	the \parspace\ of Cauchy reals with the size function as parameter.
	In particular, the polytime computable points of the interval representation are the usual polytime computable real numbers.

	Section \ref{sec: parspace of functions} introduces a parameter
	for the space of continuous functions on the unit interval, where a function is represented by an
	interval enclosure.
	It is shown that this choice of representation and parameter is optimal
	in the sense that the resulting structure of \parspace\ on the space of continuous
	functions	is minimal amongst those structures which render evaluation polytime computable.
	In particular, the polytime computable points of this representation are the usual polytime computable functions.

	Section \ref{sec:comparison to represented spaces} compares the approach presented here to the one of Kawamura and Cook.
	\Cref{resu: Characterisation of Kawamura-Cook space} characterises those \parspace{s} which admit a polytime equivalent
	length-monotone representation and thus can be treated within the framework of Kawamura and Cook.
	Section \ref{sec:comparison} shows that the interval representation of continuous functions
	is not of this kind.
	Hence, the interval representation of continuous functions
	contains strictly less information than the
  function space representation which is used by Kawamura and Cook.
	This result mainly relies on an auxiliary result due to Brau\ss e and Steinberg \cite{Aminimal}.

	\Cref{sec: non-monotone enclosures} discusses some alternative choices of representations
	to demonstrate the robustness of our definitions.
	The representation for real numbers used in \Cref{sec: non-monotone enclosures}
	 also coincides with the one used by M\"uller to model the behaviour of
	 \irram \cite{Mueller2001}.

	\subsection{Second-order complexity theory}\label{sec:second-order complexity theory}
		Fix a non-empty alphabet $\Sigma$.
		Let $M^?$ be an oracle machine,
		that is, a Turing machine with two designated tapes called the \lq oracle query\rq\ and the \lq oracle answer\rq\ tape and a special state called the \lq oracle state\rq.
		An oracle for such an oracle machine is an element of Baire space $\B:=\Sigma^{*\Sigma^*}$.
		The oracle machine $M^?$ can be executed on a given string $\str a \in \Sigma^*$ with a given oracle $\varphi \in \B$.
		It is executed like a regular Turing machine, except that whenever it enters the oracle state,
		the current content of the oracle answer tape is replaced with $\varphi(\str b)$,
		where $\str b$ is the current content of the oracle query tape.
		If the oracle machine $M^?$ terminates on input $\str a$ with oracle $\varphi$,
		the string that is written on its output tape after termination is denoted by $M^\varphi(\str a)$.
		This defines a partial function $M^\varphi$ which maps a string $\str a$
		such that $M^?$ terminates on $\str a$ and with oracle $\varphi$ to the string $M^\varphi(\str a)$.
		Every oracle machine $M^?$ computes a partial operator $F_{M^?}\colon{\subseteq\B\to\B}$ on Baire space:
		The domain of $F_{M^?}$ consists of all oracles $\varphi\in\B$ such that $M^\varphi$ is total.
		If $\varphi$ is an element of the domain of $F$ then
		the value of $F$ in $\varphi$ is given by $F_{M^?}(\varphi):= M^\varphi$.

		This is slightly different from the definition of oracle Turing machine in classical computability theory,
		where the oracles are subsets of $\Sigma^*$.
		While this is not important for computability considerations, it does make a difference when it comes to computational complexity.
		To be able to reason about complexity in this extended setting it is necessary to fix a convention for how to count oracle calls for the time consumption of a machine.
		We adapt the most common convention:
		An oracle call is considered to take one time step in which the entire answer appears on the oracle answer tape.
		Essentially this means that the machine is not forced to read the whole oracle answer.
		Another detail is the position of the read/write heads.
		We assume that the head position does not change during the oracle interaction.
		Denote the number of steps the oracle machine $M^?$ takes on input $\str a$ and with oracle $\varphi$ by $\Time_{M^?}(\varphi,\str a)$.
		Of course, one of the sanity checks for any reasonable computational model is that the details fixed above should be irrelevant.
		One should end up with the same computational complexity classes when,
		for instance,
		the machine returns its head to the beginning of the oracle query tape during the oracle interaction.

		Since the oracle $\varphi$ is considered an input of the computation, a function bounding the time consumption of an oracle machine should be allowed to depend on the size of the oracle in addition to the size of the input string.
		The most common way to measure the size of a string function is to use the worst-case length-increase from input to output.
		That is, for a string function $\varphi$ let its \demph{size} $\length{\varphi}:\NN\to\NN$ be defined by
		\begin{equation*}
			\tag{s}\label{eq:length} \length{\varphi}(n):= \max_{\length{\str a}\leq n}\{\length{\varphi(\str a)}\}.
		\end{equation*}
		Since a time bound for an oracle machine should produce a bound on the number of steps the execution takes from a size of the oracle and a size of the input it should be of type $\NN^\NN\times\NN\to\NN$.
		To talk about polytime computability it is necessary to find a subclass of functions of this type that is considered
		to have ``polynomial'' growth.
		\begin{definition}
			The class of \demph{second-order polynomials} $\sop \subseteq \NN^{\NN^\NN\times \NN}$ is the smallest class such that the following conditions hold:
			\begin{itemize}
				\item For all $p\in\NN[X]$ we have $(l,n)\mapsto p(n)\in\sop$.
				\item Whenever $P\in\sop$ then also $(l,n)\mapsto l(P(l,n))\in\sop$.
				\item Whenever $P,Q\in\sop$ then both the point-wise sum $P+Q$ and the point-wise product $P\cdot Q$ are contained in $\sop$.
			\end{itemize}
		\end{definition}
		Since second-order polynomials are used as running time bounds,
		only the values on functions that turn up as sizes of string functions are relevant.
		These are exactly the non-decreasing functions, \ie functions $l$ satisfying
		$ l(n+1)\geq  l(n)$.
		This restriction is important, as the following lemma fails for more general arguments:
		\begin{lemma}[Monotonicity]\label{resu:monotonicity}
			Let $P$ be a second-order polynomial and let $l,k:\NN\to\NN$ be non-decreasing functions such that $l$ is point-wise bigger than $k$.
			Then
			\[ \forall n\in\NN\colon P(l,n)\geq P(k,n). \]
		\end{lemma}
		The proof is a straightforward induction.

		\begin{definition}\label{def:polytime functionals}
			We call a partial operator $F\colon{\subseteq \B\to\B}$ \demph{polytime computable}, if there is an oracle machine $M^?$ that computes $F$ and a second-order polynomial $P$ such that
			\[ \forall \varphi\in\dom(F),\forall \str a\in\Sigma^*\colon \Time_{M^?}(\varphi,\str a)\leq P(\length{\varphi},\length{\str a}). \]
		\end{definition}
		It was proved by Kapron and Cook \cite{MR1374053} that
		a total operator $F:\B\to \B$ is polytime computable if and only
		if the corresponding functional $(\varphi,\str a)\mapsto F(\varphi)(\str a)$
		is basic feasible in the sense of Mehlhorn \cite{MR0411947}.
		The generalisation to partial operators adds an additional choice:
		By our choice the machine is only required to comply with the time bound in the case where the oracle is in the domain of the operator.
		Another approach would be to require the existence of a total extension that runs in polynomial time.
		This corresponds to replacing the quantification over $\dom(F)$ by a quantification over all of Baire space.
		It is possible to prove that this does indeed lead to a more restrictive notion of polytime computability \cite{kawamura_et_al:LIPIcs:2017:7737}.

		To show closure of the class of polytime computable operators under composition, one needs monotonicity from \Cref{resu:monotonicity} and the following closure property of the second-order polynomials:
		\begin{proposition}\label{resu:closure sop}
			Whenever $P$ and $Q$ are second-order polynomials, then so are the following mappings:
			\[ (l,n)\mapsto P(l,Q(l,n)) \quad\text{and}\quad P(Q(l,.),n). \]
		\end{proposition}
		Just like \Cref{resu:monotonicity}, \Cref{resu:closure sop} can be proven by a straightforward induction on the structure of second-order polynomials.
		\begin{theorem}[Composition]
			Whenever $F$ and $G$ are polytime computable, then so is $F\circ G$.
		\end{theorem}
		This is proven in a more general setting in \Cref{resu:closure under composition} and we refrain from restating the proof here.

		To compute on more general spaces representations are used:
		\begin{definition}
			Let $X$ be a set.
			A \demph{representation $\xi$ of $X$} is a partial surjective function $\xi\colon{\subseteq \B\to X}$.
			A \demph{represented space} is a pair $\XX=(X,\xi)$ of a set and a representation of that set.
		\end{definition}
		The elements of $\xi^{-1}(x)$ are called the \demph{names} of $x$.
		An element of a represented space is called \demph{computable} if it has a computable name.

		Computability of functions between represented spaces can be defined via realisers.
		\begin{definition}\label{def: realiser}
			Let $f\colon\XX\to\YY$ be a function between represented spaces.
			A function $F:\subseteq \B\to \B$ is called a \demph{realiser} of $f$ if it translates names of the input to names of the output, that is if
			\[ \varphi \in\dom(\xi_{\XX}) \Rightarrow \xi_{\YY}(F(\varphi))= f(\xi_{\XX}(\varphi)). \]
		\end{definition}
                \begin{figure}
                  \hfill
                  \begin{subfigure}[b]{.2\textwidth}
                    \begin{xy}
                    \xymatrix{
                      \XX \ar[r]^f & \YY \\
                      \B \ar[u]^{\xi_{\XX}} \ar[r]_F & \B \ar[u]_{\xi_{\YY}}}
                    \end{xy}
                    \caption{$F$ realises $f$}\label{fig: diagram}
                  \end{subfigure}
                  \hfill
                  \begin{subfigure}[b]{.65\textwidth}
                    \begin{tikzpicture}
                      \draw (0,0) ellipse (.75cm and 1cm);
                      \node at (0,.5) {$\dom(g)$};
                      \node at (-.2,0) {$x$};
                      \draw[fill=black] (0,0) circle (.02cm);
                      \draw[dotted] (0,.01) -- (2.9,.6);
                      \draw[dotted] (0,-.01) -- (2.9,-.2);
                      \node at (-.2,-.5) {$y$};
                      \draw[fill=black] (0,-.5) circle (.02cm);
                      \draw[dotted] (0,-.49) -- (2.6,-.35);
                      \draw[dotted] (0,-.51) -- (2.5,-1.345);
                      \node at (1.25,-.25) {$g$};
                      \node at (4.5, -.5) {$f$};
                      
                      \draw (3,0) ellipse (1cm and 1.5cm);
                      \node at (3,1) {$\dom(f)$};
                      \draw (3,.2) ellipse (.5cm and .4cm);
                      \node at (3,.2) {$g(x)$};
                      \draw[dotted] (3,.6) -- (5.8,1.2);
                      \draw[dotted] (3,-.2) -- (5.8,-.2);
                      \draw (2.6,-.85) ellipse (.4cm and .5cm);
                      \node at (2.6,-.85) {$g(y)$};

                      \draw (6,.5) ellipse (1.1cm and .7cm);
                      \node at (6,.5) {$(f \circ g) (x)$};
                      \node at (6, -1) {$(f \circ g)(y) = \emptyset$};
                    \end{tikzpicture}
                    \caption{Multifunction composition visualized}
                  \end{subfigure}
                  \hfill
                  \caption{
                    The diagram (a) need not commute. The function $F$ realises $f$ if $\xi_\YY\circ F$ extends $f\circ\xi_\XX$ (or tightens it if $f$ is multivalued).}
                \end{figure}
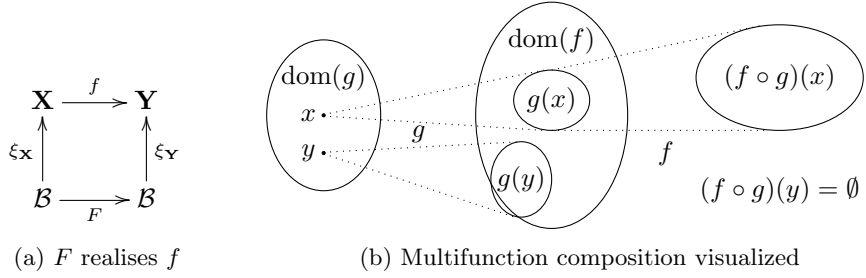
                (also compare \Cref{fig: diagram}.)
	        A function between represented spaces is called \demph{computable} if it has a computable realiser.
		It is called \demph{polytime computable} if it has a polytime computable realiser.

		The later parts of this paper need a slight generalisation of the above to
		multi-valued functions.
		Recall that a \demph{multifunction} $f\colon X\mto Y$ assigns to each element $x\in X$ a subset $f(x)\subseteq Y$.
		Its \demph{domain} $\dom(f)$ consists of all $x \in X$ whose image under $f$ is nonempty.
		A partial single-valued function $g\colon \subseteq X \to Y$ can be identified
		with the multifunction 
		which sends elements $x \in \dom (g)$ to the singleton $\{g(x)\}$ and elements
		outside of the domain of $g$ to the empty set.
		The definition of computability using realisers generalises to multifunctions in a straightforward way:
		A function $F \colon {\subseteq \B \to \B}$ is called a realiser of $f$
		if $\xi_{\YY}(F(\varphi)) \in f(\xi_{\XX}(\varphi))$ for all
		$\varphi \in \dom(\xi_{\XX})$.
		The elements of $f(x)$ are thus interpreted as
		``acceptable return values'' 
		for an algorithm which computes $f$.
		If $f \colon X \mto Y$ and
		$g \colon Y \mto Z$ are multifunctions then their composition
		$g \circ f$ is the multifunction with
		\[\dom (g \circ f) = \left\{x \in X \mid f(x) \subseteq \dom (g) \right\}\]
		and
		\[g \circ f (x) = \left\{z \in g(y)  \mid y \in f(x) \right\}. \]
		This definition ensures that the composition of two realisers is a
		realiser of the composition.
		Note that while multivalued functions can formally be identified with relations,
		from conceptual point of view it is better not do do so.
		For instance, the composition defined above is different from the natural notion of composition for relations.
		Multifunctions are a standard tool in computable analysis to avoid
		certain kinds of continuity issues and are needed in \Cref{sec:comparison} for this exact reason.

	\subsection{Notations and complexity on the reals}\label{sec:notations and complexity on the reals}
		A name of an element of a represented space should be understood as a
		black box that provides on-demand information about the object it encodes.
		For real numbers, for instance, a reasonable query to such a black box
		could be \lq provide me with a $2^{-n}$ approximation to the real number\rq, and the answer that the name provides should be such an approximation.
		The input and the output of the name are finite binary strings, and questions like the above can be formulated by encoding elements of discrete structures like the integers and the rational numbers.

		A \demph{notation of a space $X$} is a partial surjective mapping $\nu_X\colon{\subseteq\Sigma^*\to X}$.
		Fix the following standard notations:
		Let $\nu_{\Z}$ be the mapping defined on a string $\str a = a_0\sdone a_2\ldots a_{\length{\str a}}$ whose second digit is a $\sdone$ by
		\[ \nu_{\Z}(\str a) =(-1)^{a_0} \sum_{i=1}^{\length{\str a}-1} a_i 2^{\length{\str a}-i} \]
		and zero on strings that do not have a second digit.
		A dyadic rational is a rational number of the form $\frac r{2^n}$ for some $r\in\Z$ and $n\in \NN$.
		The set of these numbers is denoted by $\DD$.
		The reason for their use is that they are a good model for machine numbers,
		as they are precisely those rational numbers which have a finite binary expansion.
		Encode a dyadic number as its unique finite binary expansion starting in a code of an integer followed by a separator symbol that is used to mark the position of the decimal point.
		To avoid confusion with unary and binary notations, we do not specify a notation for the natural numbers.
		Instead of working on $n\in\NN$ directly, we use the integer $2^n$.
		This means that implicitly use the unary encoding of natural numbers while we use the binary encoding for the integers.

		To also be able to accept or return pairs of integers or dyadic numbers use a pairing function for strings.
		For technical reasons that become apparent in \Cref{sec: interval reals}, we choose to use a very specific pairing function.
		For two strings $\str a$ and $\str b$ let the pairing $\langle\str a,\str b\rangle$ be the string which is constructed as follows:
		Let $\str c$ be the string that that starts in the digit $\sdone$, repeated $\min\{\length{\str a},\length{\str b}\}$ times, followed by a $\sdzero$, then a bit indicating which of $\str a$ and $\str b$ is longer and finally enough $\sdzero$'s to make it as long as the longer of the two strings.
		Then pad the strings $\str a$ and $\str b$ to the length of the longer of the two strings by adding zeros to the end.
		$\langle \str a,\str b\rangle$ is the string whose bits alternate between the digits of $\str c$, $\str a$ and $\str b$.
		It is important for this paper that the $n$ initial segment of $\str a$ and $\str b$ can be read from a $3(n+2)$ initial segment of $\langle \str a,\str b\rangle$.
                \begin{figure}
                  \centering
                  \begin{tikzpicture}
                    \node at (0, 1) {$\str a$};
                    \node at (1.25,1) {$\sdzero$};
                    \draw[->] (1.25,.25) -- (1.25,.75);
                    \draw[->, dotted] (2.25,.25) -- (1.5,.75);
                    \node[green] at (2,.75) {\ding{51}};

                    \node at (2.75,1) {$\sdzero$};
                    \draw[->] (2.75,.25) -- (2.75,.75);
                    \draw[->, dotted] (3.75,.25) -- (3,.75);
                    \node[green] at (3.5,.75) {\ding{51}};
                    
                    \node at (4.25,1) {$\sdzero$};
                    \draw[->] (4.25,.25) -- (4.25,.75);
                    \draw[->, dotted] (5.25,.25) -- (4.5,.75);
                    \node[green] at (5,.75) {\ding{51}};

                    \node at (5.75,1) {$\sdzero$};
                    \draw[->] (5.75,.25) -- (5.75,.75);
                    \draw[->, dotted] (6.75,.25) -- (6,.75);
                    \node[orange] at (6.5,.75) {?};

                    \node[red] at (5.75,1) {\ding{55}};
                    \draw[->, dotted] (8.25,.25) -- (8.25,1) -- (6,1);
                    \node[red] at (7.5,.75) {\ding{55}};

                    \node at (0, 0) {$\langle \str a, \str b\rangle$};
                    \draw (1,.25) -- (10,.25) -- (10,-.25) -- (1,-.25) -- (1,.25);
                    \draw[thin,gray] (1.5,.25) -- (1.5,-.25);
                    \draw[thin,gray] (2,.25) -- (2,-.25);
                    \draw (2.5,.25) -- (2.5,-.25);
                    \draw[thin,gray] (3,.25) -- (3,-.25);
                    \draw[thin,gray] (3.5,.25) -- (3.5,-.25);
                    \draw (4,.25) -- (4,-.25);
                    \draw[thin,gray] (4.5,.25) -- (4.5,-.25);
                    \draw[thin,gray] (5,.25) -- (5,-.25);
                    \draw (5.5,.25) -- (5.5,-.25);
                    \draw[thin,gray] (6,.25) -- (6,-.25);
                    \draw[thin,gray] (6.5,.25) -- (6.5,-.25);
                    \draw (7,.25) -- (7,-.25);
                    \draw[thin,gray] (7.5,.25) -- (7.5,-.25);
                    \draw[thin,gray] (8,.25) -- (8,-.25);
                    \draw (8.5,.25) -- (8.5,-.25);
                    \draw[thin,gray] (9,.25) -- (9,-.25);
                    \draw[thin,gray] (9.5,.25) -- (9.5,-.25);

                    \node at (1.25,0) {$\sdzero$};                    
                    \node at (1.75,0) {$\sdone$};
                    \node at (2.25,0) {$\sdone$};
                    \node at (2.75,0) {$\sdzero$};
                    \node at (3.25,0) {$\sdzero$};
                    \node at (3.75,0) {$\sdone$};
                    \node at (4.25,0) {$\sdzero$};
                    \node at (4.75,0) {$\sdzero$};
                    \node at (5.25,0) {$\sdone$};
                    \node at (5.75,0) {$\sdzero$};
                    \node at (6.25,0) {$\sdone$};
                    \node at (6.75,0) {$\sdzero$};
                    \node at (7.25,0) {$\sdzero$};
                    \node at (7.75,0) {$\sdone$};
                    \node at (8.25,0) {$\sdzero$};
                    \node at (8.75,0) {$\sdzero$};
                    \node at (9.25,0) {$\sdzero$};
                    \node at (9.75,0) {$\sdzero$};

                    \node at (0, -1) {$\str b$};
                    \node at (1.75,-1) {$\sdone$};
                    \draw[->] (1.75,-.25) -- (1.75,-.75);
                    \draw[->, dotted] (2.25,-.25) -- (1.9,-.75);
                    \node[green] at (2.25,-.65) {\ding{51}};

                    \node at (3.25,-1) {$\sdzero$};
                    \node at (3.25, -.5) {$\cdots$};
                    \node at (4.75,-1) {$\sdzero$};
                    \node at (6.25,-1) {$\sdone$};
                    \draw[->] (6.25,-.25) -- (6.25,-.75);
                    \draw[->, dotted] (6.75,-.25) -- (6.4,-.75);
                    \node[orange] at (6.75,-.65) {?};
                    \node at (7.75,-1) {$\sdone$};
                    \draw[->, dotted] (8.25,-.25) -- (7,-.6);
                    \node[red] at (6.75,-.65) {\ding{55}};
                    \node[green] at (6.85,-.4) {\ding{51}};
                    \node at (9.25,-1) {$\sdzero$};                   
                  \end{tikzpicture}
                  \caption{The pairing $\langle \str a, \str b\rangle$ of the strings $\str a := \sdzero\sdzero\sdzero$ and $\str b := \sdone\sdzero\sdzero\sdone\sdone\sdzero$ as an example}
                \end{figure}

		In the following we use these encodings to identify Baire space with the space of functions between the encoded structures.
		For instance the statement \lq$\varphi:\NN\to\DD$ is a name of an element\rq\ is used as an abbreviation of the statement \lq any function $\psi:\Sigma^*\to \Sigma^*$ such that $\nu_{\Z}(\str a) = 2^n$ implies $\nu_{\DD}(\psi(\str a))=\varphi(n)$ is a name of the element\rq.

		\begin{definition}\label{def:cauchy reals}
			Define the \demph{Cauchy representation $\xi_{\RR_c}$ of $\RR$} as follows:
			A function $\varphi:\NN \to \DD$ is a name of a real number $x$ if and only if $\abs{\varphi(n) - x} \leq 2^{-n}$ holds for all $n\in\NN$.
		\end{definition}
		This adopts the widespread convention used in real complexity theory to provide accuracy requirements as natural numbers in unary.
		It would have equivalently been possible to provide a natural number $n$ in binary and require the return value to be a $\frac1{n+1}$-approximation or to provide a dyadic rational $\varepsilon$ and require the return value to be an $\varepsilon$-approximation.
		We refer to the space $\RR_{c}:=(\RR,\xi_{\RR_c})$, as the \demph{represented space of Cauchy reals}.
		The representation $\xi_{\RR_c}$ is used throughout literature with great confidence that it induces the right notion of complexity for real numbers and there are many results supporting this:
		The functions that have a polytime computable realiser are exactly those that
		are polytime computable in the sense of Ko \cite{MR1137517}
		as proved by Lambov \cite{MR2275414}.
		It is well known that Ko's notion can be reproduced in Weihrauch's type two theory of effectivity \cite{MR1795407}.

		While the Cauchy representation is in principle straightforward to realise on a physical computer,
		the inherent laziness of the datatype leads to undesirable memory overheads.
                
		To illustrate this, consider the task of computing the iterations of the logistic map, \ie the 
		$n^\text{th}$ number of the sequence recursively defined by
        \[
        x_0  := \tilde x, \quad\text{and}\quad x_{i+1}:= r x_i (1 - x_i),
		\]
		where $r$ and $\tilde x$ are real numbers.
		An algorithm for computing $x_n$ can be written in imperative-style pseudo-code as follows:
		\begin{align*}
			&x \leftarrow \tilde x \\
			&\text{for }  i \text{ in } 1 \text{ .. } n: \\
			&\hspace{1em}x \leftarrow r x (1 - x)
		\end{align*}
		If $x$ is taken to be, say, a floating point number or an interval with fixed precision endpoints,
		then the memory consumption of this program is essentially constant in $n$
		and linear in the number of bits used to encode $x$.
		
		Now imagine that $x$ is implemented as a Cauchy real instead.
		Then $x$ is given as a function and therefore cannot be encoded with finitely many bits.
		Hence, the straightforward way to execute the above program in this case would be to first build the computation tree (compare \Cref{fig: lazy exponential}).
        The evaluation algorithm would then proceed by 
        propagating accuracy requirements from the root to the leaves,
        computing an approximation of each leaf to the respective required accuracy,
        and finally evaluate the tree bottom-up using these approximations.
        Not taking into account the extensive recomputation of approximations that the second step may lead to, 
        the construction of the tree alone can lead to exponential memory consumption in $n$ when done naively.
		This exponential overhead can be avoided by identifying identical subtrees,
		making the tree into a Directed Acyclic Graph (DAG) (compare \Cref{fig: lazy linear})
		of linear size in $n$.
		Implementing this identification of subtrees is non-trivial in general,
		and it still leaves us with a linear memory overhead.
		For this reason, implementations like \irram use a different evaluation strategy:
		Guess an accuracy requirement and approximate all real numbers involved in the program to that accuracy.
		Use interval arithmetic to evaluate the program.
		If the end result is not sufficiently accurate, rerun the program with higher accuracy.
		This entirely avoids the memory overhead incurred from constructing the DAG.
		
		While this approach comes with its own drawbacks, such as recomputation and frequent overestimation of the needed precision, software based on interval computation is empirically superior in speed and memory consumption to implementations based on the Cauchy representation.
		
		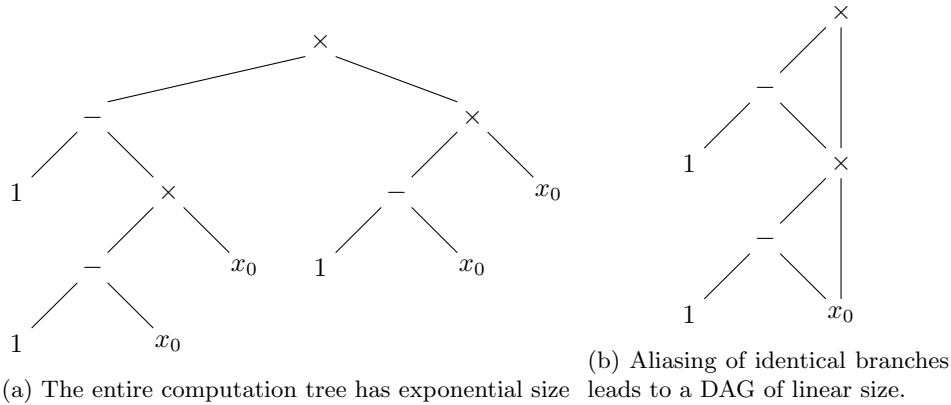
\begin{figure}
                  \hfill
	          \begin{subfigure}[b]{.6\textwidth}
		    \centering
	            \begin{tikzpicture}
	  	      \node at (0,0) {$\times$};
		      \draw[-] (0.2,-0.2) --  ( 1.8, -0.8);
		      \draw[-] (-0.2,-0.2) -- (-2.8, -0.8);
	  	      
		      \node at (-3,-1) {$-$};
		      \draw[-] (-3.2,-1.2) -- (-3.8,-1.8);
		      \draw[-] (-2.8,-1.2) -- (-2.2,-1.8);
		      \node at (-4,-2) {$1$};
		      \node at (-2,-2) {$\times$};
		      \draw[-] (-1.8,-2.2) -- (-1.2,-2.8);
		      \draw[-] (-2.2,-2.2) -- (-2.8,-2.8);
		      \node at (-1,-3) {$x_0$};
		      \node at (-3,-3) {$-$};
		      \draw[-] (-3.2,-3.2) -- (-3.8,-3.8);
		      \draw[-] (-2.8,-3.2) -- (-2.2,-3.8);
		      \node at (-4, -4) {$1$};
		      \node at (-2, -4) {$x_0$};
		      
		      \node at (2,-1) {$\times$};
		      \draw[-] (2.2,-1.2) -- (2.8,-1.8);
		      \draw[-] (1.8,-1.2) -- (1.2,-1.8);
		      \node at (3,-2) {$x_0$};
		      \node at (1,-2) {$-$};
		      \draw[-] (0.8,-2.2) -- (0.2,-2.8);
		      \draw[-] (1.2,-2.2) -- (1.8,-2.8);
		      \node at (0,-3) {$1$};
		      \node at (2,-3) {$x_0$};
	            \end{tikzpicture}
	            \caption{The entire computation tree has exponential size}
	  	    \label{fig: lazy exponential}
	          \end{subfigure}
                  \hfill
	          \begin{subfigure}[b]{.375\textwidth}
	            \centering
	            \begin{tikzpicture}
	              \node at (0,0) {$\times$};
	              \node at (-1,-1) {$-$};
	              \node at (-2,-2) {$1$};
	              \node at (0,-2) {$\times$};
	              \node at (-1,-3) {$-$};
	              \node at (-2,-4) {$1$};
	              \node at (0,-4) {$x_0$};
	              
	              \draw[-] (-0.2,-0.2) -- (-0.8,-0.8);
	              \draw[-] (0,-0.2) -- ( 0,-1.8);
	              \draw[-] (-1.2,-1.2) -- (-1.8,-1.8);
	              \draw[-] (-0.8,-1.2) -- ( -0.2,-1.8);
	              \draw[-] (-0.2,-2.2) --  (-0.8,-2.8);
	              \draw[-] (0,-2.2) --  ( 0,-3.8);
	              \draw[-] (-1.2,-3.2) -- (-1.8,-3.8);
	              \draw[-] (-0.8,-3.2) -- ( -0.2,-3.8);
	            \end{tikzpicture}
	            \caption{Aliasing of identical branches leads to a DAG of linear size.}
	            \label{fig: lazy linear}
	          \end{subfigure}
                  \hfill
	          \caption{Computing a DAG for two iterations of the logistic map with $r = 1$.}
		\end{figure}
		
		This leads to a different choice of real number representation:
		Let $\ID$ denote the set of finite dyadic intervals together with the infinite interval $[-\infty,\infty]$.
		We use the abbreviation
		\[ [r\pm\varepsilon] := [r-\varepsilon,r+\varepsilon]. \]
		Any finite dyadic interval can be written in this form and we encode such an interval as the pair of a code of $r$ and a code of $\varepsilon$ as dyadic numbers.
		Denote the length of a dyadic interval by $\diam([r\pm\varepsilon]) := 2\varepsilon$.

		\begin{definition}\label{def: interval representation}
			A function $\varphi\colon \NN \to \ID$ is a $\xi_{\RR_i}$-name of $x \in \RR$
			if $(\varphi(n))_n$ is a nested sequence of intervals with $\{x\} = \bigcap_{n \in \NN} \varphi(n)$.
		\end{definition}

		In certain implementations, notably in \irram \cite{iRRAM},
		the monotone convergence assumption of \Cref{def: interval representation} is relaxed to convergence in the Hausdorff metric.
		A more thorough discussion of this can be found in \Cref{sec: non-monotone enclosures}, where a proof is given that this choice makes no difference up to polytime equivalence.
		From a theoretical point of view it is much more convenient to work with monotone sequences of intervals.

		The use of the interval representation is avoided in real complexity theory since it does not seem to lead to a good notion of complexity:
		Every real number has names that keep the sequence of intervals constant for an arbitrary long time before decreasing the size of the next interval and these names are of slowly increasing size.
		As a consequence, the represented space has very pathological complexity theoretical properties:
		On the one hand a function operating on names of this kind may need to read a very long initial segment before having any information about the encoded object available, while not being granted any time due to the small size of the input.
		As a consequence there are usually very few polytime computable functions
		whose domain is the space of real numbers endowed with the interval representation.
		On the other hand, a function that has to produce an interval name of a real number may delay the time until it returns information about the function value indefinitely.
		Consequentially, all computable functions with values in the real numbers with the interval representation are computable in linear time \cite{MR2090390,MR3219039,SteinbergPhD}.

		The goal of the next section is to give a definition of computational complexity
		for spaces like the space of real numbers equipped with the interval representation which avoids such pathological behaviour.

	\section{Parametrised spaces}\label{sec:parametrised spaces}
		\begin{definition}
			Let $\xi:\B\to X$ be a representation.
			A \demph{parameter for $\xi$} is a single-valued total map $\mu$ from $\dom(\xi)$ to $\NN^\NN$ such that
			\[ \forall\varphi\in\dom(\xi),\forall n\in\NN\colon \mu(\varphi)(n+1)\geq \mu(\varphi)(n). \]
			The pair $(\xi, \mu)$ is called a \demph{parametrised representation} of $X$.
			The triple $(X,\xi,\mu)$ is called a \demph{pre\parspace}.
		\end{definition}
		The monotonicity assumption guarantees that the second-order polynomials behave as expected,
		\ie it makes it possible to use the monotonicity of second-order polynomials from \Cref{resu:monotonicity}.

		We do not make any assumptions about the computability or even continuity of the parameter here.
		This is for a few reasons:
		The first is that no assumptions of this kind are needed for the content of this paper.
		Results like the minimality from \Cref{resu:minimality} do provide a construction that works without this assumption.
		Being more restrictive in the definitions would make these results less general.
		Another reason is that we do encounter discontinuous parameters in situations we consider to be of practical relevance.
		An example is discussed in more detail in \Cref{sec: non-monotone enclosures}.
		While in this example an isomorphic space with continuous parameter can be found,
		it is easy to construct examples where this is not the case.
		Allowing discontinuous parameters hence allows us to investigate spaces
		that could otherwise not be equipped with a meaningful complexity notion.

		A familiar class of parameters are restrictions of the size function
		\[ \length{\varphi}(n) := \max\left\{\length{\varphi(\str a)}\,\big\vert\, \length{\str a}\leq n\right\}. \]
		Since the function $\length\cdot:\B\to\NN^\NN$ is total and all its values are non-decreasing, it can be used as a parameter for any representation.
		Its restriction to the domain of a representation is called the \demph{standard parameter} for the representation.
		In principle, any represented space can be made into a pre\parspace\ by equipping it with the standard parameter of its representation.
		
		We now come to the definition of computational complexity.
		We will specialise all definitions immediately to second-order polynomial time.
		While it is in principle straightforward to consider other classes of second-order 
		resource bounds to obtain different complexity classes,
		second-order polynomial time computability is arguably the only higher-order
		complexity class that is sufficiently well-established in the literature.

		\begin{definition}\label{def:parameter polytime}
			Let $(X,\xi_X,\mu_X)$ and $(Y,\xi_Y,\mu_Y)$ be pre\parspace s.
			A function $f\colon X\to Y$ is called \demph{computable in polynomial time} if there is a machine $M^?$ that computes a realiser of $f$ and two second-order polynomials $P$ and $Q$ such that the following conditions are satisfied:
			\begin{itemize}
				\item $P$ bounds the running time of $M^?$ in terms of the parameter,
				\ie for all oracles $\varphi\in\dom(\xi_X)$ and strings $\str a$ we have
				\[ \Time_{M^?}(\varphi,\str a) \leq P(\mu_X(\varphi),\length{\str  a}). \]
				\item $Q$ bounds the parameter blowup of $M^?$, that is for all $\varphi \in \dom(\xi_X)$ and all $n\in\NN$ we have
				\[ \mu_Y(M^\varphi)(n) \leq Q(\mu_X(\varphi),n). \]
			\end{itemize}
			We say that $M^?$ has polynomial running time and polynomially bounded parameter blowup with respect to the parameters.
		\end{definition}
		We often conflate $P$ and $Q$ to a single polynomial which bounds both the running time and the parameter blowup.
		Let us call a realiser $F \colon \subseteq \B \to \B$ computed by a machine as in \Cref{def:parameter polytime}
		a \demph{witness for the polytime computability of $f$}.
		In the case where both spaces come with the standard parameter
		a realiser $F$ is a witness for the polytime computability of the
		function $f$ if and only if it is polytime computable in the usual
		sense:
		The first condition coincides with the usual running time restriction and
		the second condition is automatic, as writing the output counts towards the total time consumption of the machine.

		As we make no assumption on the relation between the parameter of a preparametrised space
		and the size function, Definition \ref{def:parameter polytime} does not guarantee that
		the identity on a preparametrised space is polytime computable.
		If it is, the identity on Baire space need not be a witness for the polytime computability
		of the identity on the space.
		We hence need the following additional definition:

		\begin{definition}\label{def: parspace}
			A pre\parspace\ $\XX = (X,\xi_{\XX},\mu_{\XX})$ is called a \demph{\parspace}, if the identity function $\mathrm{id}_\XX: \XX\to\XX, x\mapsto x$ is polytime computable.
		\end{definition}
		In the case where the parameter is the standard parameter
		polytime computability of the identity is automatic
		as the identity on Baire space is a witness for its polytime computability.
		In general the parameter might not provide enough time to read all of an oracle answer
		and proving polytime computability of the identity function usually boils down
		to proving that limited information can be read from a beginning segment of
		the result of an oracle query.
		An example of this is discussed in \Cref{resu:the parametrised space of irram reals}.

		While \Cref{def: parspace} might look innocent, its implications should not be underestimated.
		It implicitly connects the parameter of the space to the size function:
		While for an arbitrary name there need not be any relation,
		the time constraint imposed on a witness of polytime computability of
		the identity function forces that the size of the name it returns
		is bounded by a second-order polynomial in the parameter of the input name.
		The application of such a witness hence constitutes a normalisation procedure
		which reduces the size of excessively large names.
		This connection is in particular important as it guarantees the stability under
		small changes in the model of computation as discussed in
		\Cref{sec:second-order complexity theory}.
		One could alternatively require from the beginning that the parameter
		be point-wise bigger than the size function,
		or that the identity on Baire space be a witness of polytime
		computability of the identity.
		While these alternatives are slightly more restrictive than our chosen
		definition,
		all three choices are essentially equivalent.

		Why we chose the above definition over these alternatives is a subtle point.
		An obvious drawback of our choice is that the stability under changes to
		the	model of computation is only true once a normalisation procedure has been applied.
		The proof that a pre\parspace\ is a \parspace\ usually relies on the details
		of the computational model and the details of the encodings of the discrete structures and pairs.
		Once this fact has been established a space can be specified that is stable
		under changes of the model and isomorphic with respect to the present model.
		Despite this somewhat peculiar property, our chosen approach
		has the advantage
		that it allows for the most natural definition of
		both the representations and the parameters that we are interested in.
		The other definitions usually force that either the size function shows up in the definition of
		the parameter or the normalisation procedure is hard-coded into the definition of the representation.
		The former can sometimes lead to waste of resources, as it allows for
		wasteful encodings,
		while the latter usually includes a non-canonical choice.
		The proof of \Cref{resu:the parametrised space of irram reals} is an
		instructive illustration of this.

		\begin{theorem}[Composition]\label{resu:closure under composition}
			Let $\XX$, $\YY$, and $\ZZ$ be \parspace s.
			If $f:\XX\to\YY$ and $g:\YY\to \ZZ$ are computable in polynomial time,
			then their composition $g\circ f:\XX\to \ZZ$ is also computable in polynomial time.
		\end{theorem}

		\begin{proof}
			Let $\mu_{\XX}, \mu_{\YY}$ and $\mu_{\ZZ}$ be the parameters of $\XX$, $\YY$ and $\ZZ$.
			Let $M^?$ be a machine that computes a realiser $G$ of $g$ in time and with parameter blow-up bounded by $P$.
			Let $N^?$ be a machine that computes a realiser $F$ of $f$ in time and with blow-up bounded by $Q$.
			A machine $_MN^?$ for computing $G\circ F$ can be obtained by replacing each oracle call of the machine $M^?$ with a subroutine that carries out the operations that $N^?$ would perform.
			To estimate the time this machine takes to run on input $\str a$ with oracle $\varphi$, first note that the steps the machine takes can be divided into the steps it takes when executing the commands from $M^?$ and the ones it takes when executing commands from $N^?$.

			The number of steps that are taken while executing the commands from $M^?$ is the same as the number of steps that $M^?$ would take on input $N^\varphi = F(\varphi)$ and therefore bounded by $P(\mu_{\YY}(F(\varphi)),\length{\str a})$.
			By the second condition of \Cref{def:parameter polytime} we have $\mu_{\YY}(F(\varphi))\leq Q(\mu_{\XX}(\varphi),.)$.
			Therefore, by the monotonicity of second-order polynomials from \Cref{resu:monotonicity} we have
			\[ \Time_{M^?}({F(\varphi)},\str a) \leq P(Q(\mu_{\XX}(\varphi),.),\length{\str a}). \]

			The number of steps $_MN^?$ takes with each execution of $N^?$ is bounded by
			$Q(\mu_{\XX}(\varphi),\length{\str b})$,
			where $\str b$ is the content of the tape that replaces the oracle query tape of $M^?$.
			Due to the limited time available to $M^?$ to write this query, we have $\length{\str b}\leq P(Q(\mu_{\XX}(\varphi),.),\length{\str a})$.
			Thus,
			\[ \Time_{N^?}(\varphi,\str b) \leq Q(\mu_{\XX}(\varphi),P(Q(\mu_{\XX}(\varphi),.),\length{\str a})). \]

			The number of times the oracle is called in the computation of $M^?$ with oracle $N^\varphi$ and on input $\str a$ is also bounded by the time of steps the machine $M^?$ may take.
			Thus, a bound on the total number of steps that $_MN^?$ takes on input $\str a$ with oracle $\varphi$ can be obtained by multiplying the two time bounds above.
			This can be seen to be a second-order polynomial in $\mu_{\XX}(\varphi)$ and $\length{\str a}$ using the closure properties of second-order polynomials from \Cref{resu:closure sop}.

			Finally, to obtain the bound on the output parameter, note that
			\[ \mu_{\ZZ}(_MN^\varphi)(n) = \mu_{\ZZ}(M^{F(\varphi)})(n)\leq P(\mu_{\YY}(N^\varphi),n)\leq P(Q(\mu_{\XX}(\varphi),.),n). \]
			This completes the proof that $_MN^?$ computes $G\circ F$ in polynomial time.
			Since $G\circ F$ is a realiser of $g\circ f$ it follows that $g\circ f$ is polytime computable.
		\end{proof}
		\Cref{resu:closure under composition} shows that \parspace s form a category with polytime computable mappings as morphisms.
		It includes the closure of second-order polytime computable operators under composition as a special case.
		The proof of \Cref{resu:closure under composition} is considerably more uniform than its statement:
		a polytime algorithm for computing $g \circ f$ is obtained by composing
		any two polytime algorithms for $f$ and $g$ in the natural way.

		The rest of this section introduces some basic notions and constructions
		that are needed for reasoning about \parspace s throughout the paper.
		We use straightforward adaptations from the theory of represented spaces.
		The correctness of our choices is supported	by category theory in the sense that they
		are the \lq usual\rq\ ones in the category of \parspace s with polytime computable functions as morphisms.

		Real complexity theory has a history of non-uniformity and as a result the point-wise
		complexity structure is often known in more detail than the uniform structure.
		This makes it desirable to be able to reason about points of \parspace s.
		We arrive at the following notion:
		\begin{definition}\label{def:polytime points}
			An element of a pre\parspace\ is called \demph{computable in polynomial time}
			if it has a polytime computable name whose parameter is bounded by a polynomial.
		\end{definition}
		Another way of thinking about a polytime computable point in a \parspace\ is
		as a polytime computable map from the one-point space.
		Here, the one-point space is equipped with the unique total representation and the constant
		zero parameter and is the terminal object of the category of \parspace s.
		A polytime computable point is therefore what is referred to as \lq global element\rq\ in
		category theory.
		We obtain the following corollary:

		\begin{corollary}\label{resu: polytime computable function preserves polytime computable points}
			Polytime computable functions between \parspace s take polytime computable points to polytime computable points.
		\end{corollary}

		The usual construction of the product of two represented spaces can be extended to define a product of \parspace s.
		Define the \demph{pairing function} $\langle .,.\rangle:\B\times \B\to \B$ by
		\[ \langle\varphi,\psi\rangle(\str a) := \begin{cases} \epsilon &\text{if } \str a = \epsilon \\ \varphi(\str b) &\text{if } \str a = 0 \str b \\ \psi(\str b) &\text{if } \str a = 1\str b \end{cases} \]
		\begin{definition}
			Let $\XX=(X,\xi_{\XX},\mu_{\XX})$ and $\YY=(Y,\xi_{\YY},\mu_{\YY})$ be \parspace s.
			Equip the product $X\times Y$ with the representation
			\[ \xi_{\XX\times\YY}(\varphi) = (x,y) \quad\Leftrightarrow\quad \exists \psi,\tilde\psi: \varphi = \langle \psi,\tilde \psi\rangle \text{ and } \xi_{\XX}(\psi) = x \text{ and } \xi_{\YY}(\tilde\psi) = y, \]
			\ie a name of a pair is a pair of names of the components.
			Furthermore, equip this space with the parameter $\mu_{\XX\times\YY}$ defined by
			\[ \mu_{\XX\times \YY}(\langle\psi,\tilde \psi\rangle)(n) := \max\{\mu_{\XX}(\psi)(n),\mu_{\YY}(\tilde \psi)(n)\} \]
		\end{definition}
		The triple $(X\times Y,\xi_{\XX\times\YY},\mu_{\XX\times \YY})$ is denoted by $\XX\times \YY$.
		It is straightforward to see that $\XX\times \YY$ is indeed the product of $\XX$ and $\YY$ in the category of \parspace s and polytime computable functions.

		In the theory of representations the notion of reduction plays a central role.
		It generalises easily to parametrised representations:
		\begin{definition}\label{def: polytime translatability}
			Let $X$ be a set and let $(\xi_0, \mu_0)$ and $(\xi_1, \mu_1)$ be
			parametrised representations of $X$.
			We say that $(\xi_0, \mu_0)$ is \demph{polytime translatable} to
			$(\xi_1, \mu_1)$ if the map
			$(X, \xi_0, \mu_0) \to (X, \xi_1, \mu_1), \; x \mapsto x$
			is polytime computable.
			If $(\xi_0, \mu_0)$ is polytime translatable to $(\xi_1, \mu_1)$
			and $(\xi_1, \mu_1)$ is polytime translatable to $(\xi_0, \mu_0)$,
			we say that $(\xi_0, \mu_0)$ and $(\xi_1, \mu_1)$ are
			\demph{polytime	equivalent}.
		\end{definition}

		We chose the word ``translatable'' over the more common
		term ``reducible'' as this leads to less confusion about the
		direction of the translations.
		If a parametrised representation $(\xi_0, \mu_0)$ is polytime
		translatable to a parametrised representation $(\xi_1, \mu_1)$
		we also say that $(\xi_1, \mu_1)$
		\demph{contains less information than} $(\xi_0, \mu_0)$.
		This is a slight abuse of language as ``information content''
		is more appropriately measured by topological or computable translatability.
		
		If
		$(\xi_0, \mu_0)$ and $(\xi_1, \mu_1)$ 
		are polytime equivalent parametrised representations
		of $X$
		then the pre\parspace s
		$(X, \xi_0, \mu_0)$ and 
		$(X, \xi_1, \mu_1)$
		are polytime isomorphic,
		the underlying map of the isomorphism being the identity on $X$.
		As usual, we call pre\parspace s $A$ and $B$ \demph{isomorphic}
		if there exist polytime computable maps $\alpha \colon A \to B$
		and $\beta \colon B \to A$ with 
		$\alpha \circ \beta = \operatorname{id}_{B}$
		and 
		$\beta \circ \alpha = \operatorname{id}_{A}$.
		Note that polytime isomorphic pre\parspace s with the same underlying set
		have polytime equivalent parametrised representations up to renaming 
		of elements.
		We prefer to state our results in terms of isomorphism
		of \parspace s rather than in terms of equivalence
		of parametrised representations,
		although all isomorphisms we construct in this paper
		will have the identity as underlying map.

	\subsection{A \parspace\ of real numbers}\label{sec: interval reals}

	Recall the interval representation $\xi_{\RR_i}$ of the real numbers from \Cref{def: interval representation}:
	\begin{quote}
		A function $\varphi\colon \NN \to \ID$ is a $\xi_{\RR_i}$-name of $x \in \RR$
		if $(\varphi(n))_n$ is a nested sequence of intervals with $\{x\} = \bigcap_{n \in \NN} \varphi(n)$.
	\end{quote}
	Also recall that the represented space $\RR_i = (\RR,\xi_{\RR_i})$
	has very pathological complexity properties.
	This rules out the standard parameter for making this space into a \parspace.
	It is possible to endow the space with a different parameter which yields a sensible complexity theory.
	For a real number $x$, let $\lceil x \rceil$ denote the least integer number bigger than or equal to $x$ and let $\lb$ denote the binary logarithm function.

	\begin{definition}\label{def: parspace of interval reals}
		For a $\xi_{\RR_i}$-name $\varphi$ of $x \in \RR$, define the parameter $\mu_{\RR_i}$ as
		\[ \mu_{\RR_i}(\varphi)(n) = \min\{N\mid \diam(\varphi(N))\leq 2^{-n} \} + \lceil\lb(\abs{x} + 1)\rceil. \]
		The \demph{\parspace\ of interval reals} is the triple
		\[\RR_i = \left(\RR, \xi_{\RR_i}, \mu_{\RR_i} \right). \]
	\end{definition}

	The parameter mainly encodes the rate of convergence of a sequence of intervals.
	Small parameter blowup for a realiser of a function
	$f \colon \RR_i \to \RR_i$
	hence means that the rate of convergence of the output sequence is similar to the rate of
	convergence of the input sequence.
	It remains to show that this really defines a \parspace, \ie
	that the parameter is well-defined on the domain of the representation
	and that the identity is polytime computable.

	\begin{proposition}\label{resu:the parametrised space of irram reals}
			The space $\RR_{i}$ is a \parspace.
	\end{proposition}

	\begin{proof}
			That the parameter is well-defined follows directly from the definitions.

			A family of witnesses of the polytime computability of the identity on
			the space $\RR_{i}$ can be specified as follows:
			For a fixed non-constant polynomial $p\in\NN[X]$, let $M_{p}^?$ be the oracle machine
			that on input $n\in \NN$ (as usual encoded in unary) and with oracle $\varphi$
			queries the oracle for the interval $[r\pm\varepsilon]:=\varphi(n)$.
			If the interval is infinite, it returns the infinite interval.
			Otherwise it reads approximations $r'$ and $\varepsilon'$ to precision $2^{-p(n)-1}$ from  initial segments of $r$ and $\varepsilon$ to
			compute numbers $a_n$ and $b_n$ such that $a_n$ is the largest dyadic number with denominator $2^{p(n)}$
			with $a_n < r' - \varepsilon'$ and $b_n$ is the smallest dyadic number with denominator $2^{p(n)}$
			with $b_n > r' + \varepsilon'$.
			It then returns the interval $[a_n, b_n]$.
			To see that each of the machines $M^?_{p}$ computes a witnesses of
			the polytime computability of the identity on $\RR_{i}$,
			first note that it computes the identity:
			By construction we have $[r\pm \varepsilon]\subseteq [a_n,b_n]$.
			Also by construction, the resulting sequence of intervals is nested.
			Thus any of the intervals returned by $M^\varphi_p$ contains the real number that $\varphi$ encodes.
			To see that the diameter of the intervals still goes to zero let some $\delta>0$ be given.
			Since $\varphi$ is a name, there exists an $N$ such that for all $n$ bigger
			than $N$ the corresponding $\varepsilon$ is smaller than $\frac\delta3$.
			The polynomial $p$ is non-constant and therefore $p(n)\geq n$ holds for all $n\neq 0$.
			Thus we may choose $n$ so big that $\abs{b_n - a_n} \leq 2^{p(n)+1} + 2 \varepsilon\leq \delta$.

			To obtain a polynomial bound on the running time of $M^?_p$,
			first note that the time that $M^\varphi_p$ takes on input of length $n$
			can be bounded in terms of $p(n)$ and a bound on the absolute value of $\xi_{\RR_i}(\varphi)$.
			Furthermore,
			\[ \mu_{\RR_i}(M^\varphi_p)(n) = \min\Set{N}{\diam(M^\varphi_p(m))\leq 2^{-n}} \]
			and
			\[ \diam(M^\varphi_p)(m) \leq \diam(\varphi(m)) + 2^{-p(m)+1}. \]
			This, together with $p(m)\geq m$ implies
			\[ \mu_{\RR_i}(M^\varphi_p)(n) \leq \mu_{\RR_i}(\varphi)(n) + n +1. \]
			Since the right hand side is a second-order polynomial, this proves that
			$M^?_p$ has polynomially bounded parameter blowup.
	\end{proof}

		Choosing a witness of polytime computability of the identity corresponds
		to restricting the maximal precision that may be present in the $n^\text{th}$ component of a name.
		However, the way in which the precision is restricted is mostly arbitrary
		and it may be beneficial in practice to use different cut-off precisions
		in different computations.

		We often indirectly use the polytime computability of the identity by assuming that the value of the parameter on all names of real numbers is linear.
		A machine that works correctly on names of linear parameter can be transformed into one that works correctly on all names by precomposing it with one of the realisers from the previous proof, where the corresponding polynomial is chosen linear.
		In the following we use the big-o notation:
		For integer functions $f$ and $g$ say that $f\in\bigo{g}$ if there exists a
		constant $C$ such that for all $n$ we have $f(n) \leq C g(n) + C$.

		\begin{proposition}[$\RR_i\simeq \RR_c$]\label{resu:equivalence to the Cauchy representation}
			The space of Cauchy reals and the space of interval reals are polytime isomorphic as \parspace s.
		\end{proposition}

		\begin{proof}
			First construct the translation from the Cauchy reals to the interval reals:
			Let $M^?$ return on oracle $\varphi\in\xi_{cR}^{-1}(x)$ and input $n\in\NN$ the interval $[\varphi(n)\pm2^{-n}]$.
			Then $M^\varphi$ is a $\xi_{\RR_i}$-name of $x$.
			To produce this result, the machine needs to make $\bigo{n}$ steps for copying $n$ (which is given in unary) and $\bigo{\length{\varphi}(n)+n}$ steps produce the return value from $n$ and the dyadic approximation returned by $\varphi$.
			%It remains to check the additional condition of polytime computability with respect to parameters:
			It remains to check that the machine has appropriate parameter blow-up.
			By definition of the parameter of the interval reals we have
			\[ \mu_{\RR_i}(M^\varphi)(n) = \min\Set{N\in\NN}{\diam(M^\varphi(N))\leq 2^{-n}} + \lceil\lb(\abs{x} + 1)\rceil. \]
			By the construction of $M^?$ we have $\diam(M^\varphi(n-1))=2^{-n}$, and therefore the first summand in the above equation is always bounded by $n-1$.
			The absolute value of $x$ on the other hand is bounded by the size of the encoding of the dyadic number returned by $\varphi$.
			That is
			\[ \mu_{\RR_i}(M^\varphi)(n) \leq n - 1 + \lceil\lb(\abs{x} + 1)\rceil \leq n-1 + \length{\varphi(0)} \leq n+\length\varphi(c), \]
			where $c=2$ is the length of the encoding of $0$.
			This proves that $M^?$ computes a witness of polytime computability of the translation.

			For the other direction first note that, by applying an appropriate witness of the polytime computability of the identity from \Cref{resu:the parametrised space of irram reals}, it may be assumed that the size of any $\xi_{\RR_i}$-name of some $x\in[0,1]$ satisfies $\length\varphi \in\bigo n$.
			Define a machine $N^?$ that on such an oracle $\varphi$ and on input $n$ proceeds as follows:
			It searches for an $m$ such that $\diam(\varphi(m))\leq 2^{-n+1}$.
			This condition can be checked in time $\bigo{n}$ since the name is short.
			The search halts as soon as the value of $m$ is equal to the first summand of the parameter $\mu_{\RR_i}(\varphi)(n)$.
			Let the machine return the midpoint of the interval.
			Since all names are short and the encodings reasonable, obtaining the midpoint takes at most time $\bigo{n}$.
			By construction the machine $N^?$ does at most $\mu_{\RR_i}(\varphi)(n)$ loops of a computation that takes $\bigo n$ steps to carry out and thus runs in polynomial time.
			Since the size of the return value is in $\bigo{n}$, the machine has polynomially bounded parameter blowup.
		\end{proof}

		Since polytime computable functions preserve polytime computable points by \Cref{resu: polytime computable function preserves polytime computable points} we obtain:

		\begin{corollary}
		A real number is polytime computable if and only if it is polytime computable
		as an element of the \parspace\ $\RR_i$.
		\end{corollary}

		A final property to mention about the space $\RR_i$ and a property that distinguishes it from the Cauchy reals
		and is therefore not preserved under isomorphism is that any polytime computable function on the
		interval reals can be computed by a machine whose	running time is bounded by a first-order function.

		\begin{theorem}\label{resu: polytime function on interval reals has first-order time bound}
			Whenever $f\colon{\subseteq \RR_i \to \RR_i}$ is polytime computable, then $f$ can be computed by an oracle machine $M^?$ such that there exists a $C\in\NN$ and a second-order polynomial $P$ satisfying
			\[ \Time_{M^?}(\varphi,\str a) \leq C\cdot \length{\str a} + C \quad\text{and}\quad \mu_{i}(M^\varphi)(n) \leq P(\mu_X(\varphi),n). \]
		\end{theorem}

		\begin{proof}
			Since $f$ is polytime computable, there exists a machine $N^?$ and a second-order polynomial $Q$ that bounds the running time and the parameter blowup.
			Without loss of generality assume $Q(l,n) \geq n$.
			Let the machine $M^?$ with oracle $\varphi$ and input $\str a$ spend $\length{\str a}$ steps on simulating what $N^?$ does on oracle $\varphi$ and binary encodings of $2^i$ as input for $i$ starting from zero and counting up.
			Return the return value of the machine $N^?$ on the biggest input where the simulation finished in time.
			In case none of the computations have terminated, return the infinite interval.
			Obviously, $M^?$ takes not more than $C\cdot \length{\str a}+ C$ steps, thus it is left to specify an appropriate polynomial $P$.
			For this, note that for a given $n$, since $Q$ bounds the running time of $N^?$, choosing $\length{\str a}$ bigger than
			\[ \sum_{i=0}^n Q(\mu_{\RR_i}(\varphi),i) \leq n Q(\mu_{\RR_i}(\varphi), n) \]
			forces all the simulations of $N^?$ with oracle $\varphi$ and on input smaller than $2^n$ to come to an end.
			Since the parameter blowup of $N^?$ is bounded by $Q$, the absolute value of the number encoded by $N^\varphi$ is bounded by $Q(\mu_{\RR_i}(\varphi),0)+1$ and choosing $n$ bigger than $Q(\mu_{\RR_i}(\varphi),m)$ forces the diameter of the returned interval to be smaller than $2^{-m}$.
			This implies that on input $\str a$ of size bigger than
			\[ Q(\mu_{\RR_i}(\varphi),m) \cdot  Q(\mu_{\RR_i}(\varphi), Q(\mu_{\RR_i}(\varphi),m)) \]
			the interval that $M^?$ returns has diameter smaller than $2^{-m}$ and that $P$ can be picked as
			\[ P(l,n) := Q(l,n) \cdot  Q(l, Q(l,n)) + Q(l,0)+1. \]
			That $P$ is a second-order polynomial follows from the closure properties of second-order polynomials from \Cref{resu:closure sop}.
		\end{proof}
		Note that the proof follows the construction which is used to show that
		the computational complexity of the interval representation
		is ill-behaved with respect to running time bounds in terms of the size function
		\cite{MR2090390,MR3219039,SteinbergPhD}.
		However, in contrast to the size of a name, the parameter of a name contains meaningful information and as a consequence
		delaying the time until a meaningful output is produced leads to an increase of the parameter blow-up.
		That is, instead of removing computational cost, the construction trades time needed to produce approximations
		for a worse convergence behaviour.
		While the extent to which this is done in
		\Cref{resu: polytime function on interval reals has first-order time bound}
		may not be appropriate,
		this construction
		can be viewed as a means of separating the reasoning about an algorithm
		into two parts.
		The first part is the computation of approximations to the function value
		from approximations to the input. The complexity of these computations
		can be expressed using first-order bounds.
		The second part is the convergence analysis, which does require second-order
		bounds, as it relates the rate of convergence of the input sequence
		to the rate of	convergence of the output sequence.
		This seems to be more in line with how the complexity of algorithms is studied in
		the real number model and related models.
		Most ``natural'' algorithms for computing functions
		$f \colon \RR_i \to \RR_i$
		do indeed have first-order time bounds.

	\subsection{A \parspace\ of continuous functions}\label{sec: parspace of functions}

	Let $I=[0,1]$ denote the closed unit interval.
	Let $I_i$ be the \parspace\ that is obtained by considering the unit interval as a subspace of the \parspace\ $\RR_i$ of interval reals.
	By this we mean that the representation $\xi_{I_i}$ of $I_i$ is the range restriction of $\xi_{\RR_i}$ to $I$ and the parameter $\mu_{I_i}$ is the restriction of $\mu_{\RR_i}$ to the domain of the representation.

	Consider the space  $C(I)$ of continuous functions from the unit interval to the real numbers.
	Having made $\RR$ and $I$ into \parspace s which are closely tied to the complexity of interval methods,
	it is natural to ask whether the function space $C(I)$ admits a similar structure.

	\begin{definition}\label{def: interval function representation}
		Define the \demph{interval function representation $\xi_{if}$} as follows:
		A function $\psi\colon \ID\to\ID$ is a name of $f\in C([0,1])$ if and only if
		\[ \forall \varphi\in\dom(\xi_{\RR_i})\colon\big(\xi_{\RR_i}(\varphi) \in[0,1] \quad\Rightarrow\quad \xi_{\RR_i}(\psi\circ\varphi) = f(\xi_{\RR_i}(\varphi))\big). \]
	\end{definition}

	Note that if $\psi$ is a name of a continuous function $f$, then $\psi$ is necessarily monotone as an interval function, \ie $J\subseteq I$ implies that $\psi(J)\subseteq \psi(I)$.
	Unlike the Kawamura-Cook representation of continuous functions,
	which is recalled in Section \ref{sec:comparison to represented spaces}, the present
	Definition \ref{def: interval function representation} employs a
	canonical exponential construction to	represent the function space.
	The restriction to compact intervals of reals is mainly necessary in order to
	ensure the well-definedness of the parameter,
	which essentially encodes a modulus of uniform continuity of the represented function:

	\begin{definition}\label{def:parameters for the irram functions}
			Let the \demph{parameter} $\mu_{if}\colon{\dom \left(\xi_{if}\right) \to \NN^\NN}$ be defined by
			\begin{equation*}
				\begin{split}
					\mu_{if}&(\psi)(n) :=  \lceil\lb(\norm{\xi_{if}(\psi)}_\infty + 1)\rceil \\
					&+ \min\Set{N\in\NN}{\forall J\in\ID: \diam(J)\leq2^{-N}\Rightarrow \diam(\psi(J))\leq 2^{-n}},
				\end{split}
			\end{equation*}
			whenever $\psi$ is a $\xi_{if}$-name of a function.
	\end{definition}
	While the time taken by a polytime machine operating on \parspace\ may depend on the parameter of the input name,
	the machine does not have direct access to the parameter and hence can in general not use it in its computation.
	The best way to compute a modulus of continuity from a name of a function in the interval function representation seems to be a search that may in the worst case take time exponential in the value of the parameter defined above.
	Indeed, in \Cref{sec:comparison} we show that the modulus function cannot be computed in polynomial time with respect to the above parameter.
	This is a difference to the representation used by Kawamura and Cook, where a name comes with explicit information about the modulus of continuity.

	\begin{lemma}\label{resu: interval function parameter well-defined}
	The parameter $\mu_{if}$ is well-defined on $\dom(\xi_{if})$.
	\end{lemma}
	\begin{proof}
	Assume that $\psi$ is a name of a function.
	Our goal is to show that for every $n \in \NN$ the minimum
	\[\min\{ N \in \NN \mid \forall J \in \ID\colon \diam(J) \leq 2^{-N} \Rightarrow \diam(\psi(J)) \leq 2^{-n} \} \]
	exists.
	For a fixed $x$ denote the smallest dyadic number with denominator $2^n$ that is bigger than $x$ by $x_n$.
	Then $\varphi_x(n) := [x_n\pm 2^{-n}]$ is a $\xi_{\RR_i}$-name of $x$.
	Since we assumed $\psi$ to be a name, it follows that
	$\psi(\varphi_x(n)) \to \{f(x)\}$ as $n \to \infty$.
	Hence, there exists $N_x \in \NN$ with
	$\diam(\psi(\varphi_x(N_x)))) \leq 2^{-n}$.
	The family $((x_{N_x}\pm 2^{-N_x}))_{x \in [0,1]}$ of interiors of the intervals $\varphi_x(N_x)$
	is an open cover of $[0,1]$.
	Since the unit interval is compact, there exists a finite subcover.
	As a finite family of open intervals, this family has a smallest overlap.
	Let $N$ be big enough such that $2^{-N}$ is smaller than this overlap.
	It follows that every interval $J$ whose diameter is smaller than $2^{-N}$ is contained in an interval $I$ from this finite family.
	By the monotonicity of $\psi$, for each such interval it follows that $\psi(J) \subseteq \psi(I)$ and since the diameter of $\psi(I)$ is smaller than $2^{-n}$, the same holds for $\psi(J)$.
	\end{proof}

	Consider the space $C(I)_i = (C([0,1]),\xi_{if},\mu_{if})$.
	We call this the \demph{\parspace\ of interval functions}.
	Just like the corresponding result for the interval reals,
	proving that the interval functions are a \parspace\ mainly consists of
	the introduction of a rounding procedure,
	which includes the non-canonical choice of a polynomial that controls the
	cut-off precision.
	\begin{proposition}
		The space $C(I)_i$ is a \parspace.
	\end{proposition}

	\begin{proof}
		First note that, while copying the input interval to the oracle query tape is possible in polynomial time, it may not be possible to copy the answer to the output tape.
		This is because a name of a function may return intervals $[r\pm\varepsilon]$ such that $\varepsilon$ is fairly big and $r$ is fairly small but still $r$ and $\varepsilon$ have large encodings.
		The time the machine is granted, however, only increases with the diameter of the interval getting smaller or the midpoint getting bigger.
		Instead of copying $r$ and $\varepsilon$ directly, rounded versions of these numbers can be read from a beginning segment of the oracle answer.
		The details of this rounding procedure are given in \Cref{resu:the parametrised space of irram reals}:
		The machines $M^?_p$ introduced there make a single oracle query at the very beginning of the computation
		and may therefore instead be considered as machines that do not have an oracle and take an interval as input.
		Fix such a machine $M_p$.
		A witness of polytime computability for the identity is computed by the machine which maps
		a given name $\varphi$ of a function to the composition of the machine $M_p$ and $\varphi$.
	\end{proof}

		Consider the evaluation operator defined as follows:
		\[ \eval:C([0,1])\times [0,1] \to \RR, \quad (f,x)\mapsto f(x). \]
		A sanity check for the definition of $C([0,1])_i$ is that evaluation should be polytime computable.
		For this statement to make sense, it is necessary to use the product of \parspace s given at the end of the introduction of \Cref{sec:parametrised spaces}.
		\begin{proposition}[Evaluation]\label{resu:evaluation}
			Evaluation as operator from $C(I)_i\times I_i$ to $\RR_i$ is polytime computable.
		\end{proposition}

		\begin{proof}
			Consider the machine $M^?$ that when given oracle $\langle\psi,\varphi\rangle$
		  with $\xi_{if}(\psi) = f$ and $\xi_{\RR_i}(\varphi) = x$, such that $\varphi$ has linear size returns a string function of linear size obtained from $\psi(\varphi(n))$ by applying an appropriate realiser of the identity constructed in \Cref{resu:the parametrised space of irram reals}.
		  This machine computes a realiser of the evaluation operator by the definition of the representation $\xi_{if}$ from \Cref{def: interval function representation}.
			Due to the appropriate truncations, the time the machine takes is bounded polynomially.
			To see that the machine does not inflate the parameter too much, note that
			\begin{align*}
				\mu_{\RR_i}(M^{\langle\psi,\varphi\rangle})(n) & = \min\Set{N}{\forall m\geq N\colon \psi(\varphi(m))\leq 2^{-n+1}} + \lceil{\lb(\abs{f(x)} + 1)}\rceil\\
				& \leq \mu_{if}(\psi)(\mu_{\RR_i}(\varphi)(n))\\
				& \leq (\mu_{if}\times\mu_{\RR_i})(\langle\psi,\varphi\rangle)
								\left( (\mu_{if}\times\mu_{\RR_i})(\langle\psi,\varphi\rangle)(n) \right).
			\end{align*}
			Thus $M^?$ computes a witness of the polytime computability of the evaluation operator.
		\end{proof}

		Recall the notion of polytime translatability from \Cref{def: polytime translatability}.
		The parametrised representation $(\xi_{if},\mu_{if})$ is the ``correct''
		representation for $C(I)$, viewed as a function space,
		as it contains the least amount of information
		(in the sense of \Cref{def: polytime translatability})
		among those representations which render evaluation polytime
		computable.

		\begin{theorem}[Minimality]\label{resu:minimality}
			For a parametrised representation $(\xi,\mu)$ of $C([0,1])$ the following are equivalent:
			\begin{enumerate}
				\item Evaluation is polytime computable with respect to $(\xi,\mu)$.
				\item The pair $(\xi,\mu)$ is polytime translatable to $(\xi_{if},\mu_{if})$.
			\end{enumerate}
		\end{theorem}

		\begin{proof}
			The implication {\itshape2.$\Rightarrow$1.} directly follows from the closure of the polytime computable operators under composition from \Cref{resu:closure under composition} and the polytime computability of the evaluation operator on the interval functions from \Cref{resu:evaluation}.

			For the other implication assume that the evaluation operator is polytime computable with respect to $(\xi,\mu)$.
			Note that due to the equivalence of the Cauchy reals and the interval reals from \Cref{resu:equivalence to the Cauchy representation}, the evaluation operator is also polytime computable if $[0,1]$ and $\RR$ are equipped with the Cauchy representation.
			Let $M^?$ be a machine that computes evaluation with time consumption and parameter blowup bounded by a second-order polynomial $P$.

			Define a machine $N^?$ that computes a translation of $(\xi,\mu)$ into $(\xi_{if},\mu_{if})$ in polytime as follows:
			Fix some $\xi$-name $\psi$ of some function $f\in C([0,1])$ as oracle and an interval $[r\pm \varepsilon]$
			with $r \in [0,1] \cap \DD$ as input.
			Let $n$ be the largest natural number such that $\varepsilon \leq 2^{-n}$
			(if the input interval has diameter bigger than one, return the interval $[0\pm \infty]$).
			Let $r_m$ be $r$ rounded to precision $2^{-m}$.
			The machine $N^?$ follows the steps that $M^?$ would take on oracle $\langle \psi, m\mapsto r_m\rangle$ and inputs $i$ going from $n$ to zero.
			In each of the runs it saves the maximal query $k$ that is posed to the oracle $m\mapsto r_m$ and when the computation on $i$ ends, it compares $k$ to $n$.
			If $k$ is bigger than $n$ the machine decreases $i$ and starts over, unless $i$ is already zero, in which case it returns $[0\pm\infty]$.
			If $k$ is smaller than $n$, and $M^?$ returns $d$, then let $N^?$ return the interval $[d\pm2^{-i}]$.

			To see that this produces a $\xi_{if}$-name of $f$ from the $\xi$-name $\psi$, let $P$ be the polynomial time-bound of $M^?$.
			Let $([r_s \pm \varepsilon_s])_{s \in \NN}$ be a sequence of intervals that converge to a point $x \in [0,1]$.
			Let $([d_s \pm \delta_s])_{s \in \NN}$ be the sequence of intervals which are returned by $N^?$ on input $([r_s \pm \varepsilon_s])_{s \in \NN}$.
			We may assume without loss of generality that $([d_s \pm \delta_s])_{s \in \NN}$ is a nested sequence.
			First note that each of the intervals $[d_s \pm \delta_s]$ contains $f(x)$, so that it remains to show that
			the sequence $([d_s \pm \delta_s])_{s \in \NN}$ converges to a point.
			There exists an integer constant $K$ such that
			for any $r \in [0,1]\cap \DD$ the function $m\mapsto r_m$ has size smaller than $K \cdot\left(m+1\right)$.
			Thus, the number of steps of $N^?$ in any of the simulations of $M^?$ and in particular the value of $k$ in each such simulation is bounded by
			\begin{equation*}\tag{k}\label{eq:k}
				k_n:=P(m\mapsto \max\{\mu(\psi)(m),K\cdot(m + 1)\},n).
			\end{equation*}
			Since this value is independent of $r$, the computation on all intervals of diameter smaller or equal $2^{-k_n}$ results in return values of diameter smaller than $2^{-n}$.
			In particular the sequence $([d_s \pm \delta_s])_{s \in \NN}$ converges to a point.

			Finally, the machine $N^?$ runs in polytime:
			To bound the number of steps $N^?$ takes by a second-order polynomial in the length of the input and the parameter of the oracle note that the number of steps taken in each simulation of $M^?$ is bounded by $k_n$ as above and that at most $n$ of these simulations need to be carried out.
			Let us now show that $N^?$ has polynomial parameter blowup.
			We have
			\begin{align*}
			\mu_{\RR_i}(N^\psi)(n)
				&=  \lceil\lb\left(\norm{f}_\infty + 1\right)\rceil \\
					&+ \min\Set{m\in\NN}{\forall J\in\ID: \diam(J)\leq2^{-m}\Rightarrow \diam(N^\psi(J))\leq 2^{-n}} \\
				&\leq  \lceil\lb\left(\norm{f}_\infty + 1\right)\rceil + k_n.
			\end{align*}
			Since $k_n$ is polynomially bounded in $\mu(\psi)$ by \eqref{eq:k},
			it remains to provide a bound on the supremum norm of the function $f$.
			Cover $[0,1]$ with finitely many intervals of the form
			$[r \pm 2^{-k_0}]$
			where $r$ is a dyadic rational number with $\bigo{k_0}$ bits.
			When these are fed into the machine $N^{\psi}$, it will produce approximations
			to the range of $f$ over these intervals to error $1$.
			Hence a bound on the output size of the machine over these intervals yields a bound on the supremum norm of $f$.
			By our previous considerations the running time (and hence the output size) of the machine on each interval is bounded by
			\begin{align*}
				k_{k_0} &= P(\max\{\mu(\psi), m \mapsto K\cdot(m + 1)\}, k_0) \\
								&= P(\max\{\mu(\psi), m \mapsto K\cdot(m + 1)\}, P(\max\{\mu(\psi), m \mapsto K\cdot(m + 1)\}, 0))
			\end{align*}
			so that the supremum norm of $f$ is bounded polynomially in $\mu(\psi)$.
		\end{proof}

		As is usually the case for minimality results of this kind,
		the proof generalises to the slightly stronger statement that
		$(C(I)_i,\eval)$ is an exponential object in the category of
		\parspace s.
		More explicitly:
		\begin{corollary}\label{resu: currying}
			For any \parspace\ $\ZZ$ and any polytime computable $F:\ZZ\times I_i \to \RR_i$ there exists a unique polytime computable mapping $c(F):\ZZ \to C(I)_i$ such that for all $z\in\ZZ$ and $x\in I$ we have $c(F)(z)(x) = F(z,x)$.
		\end{corollary}

		Together with the equivalence of the Cauchy and the interval reals from \Cref{resu:equivalence to the Cauchy representation} we obtain the following result:

		\begin{corollary}\label{resu: polytime computability of functions}
		The following are equivalent for a function $f \colon [0,1] \to \RR$:
		\begin{enumerate}
		\item $f$ is computable in polynomial time with respect to the Cauchy representation of $[0,1]$ and $\RR$.
		\item $f$ is computable in polynomial time with respect to the parametrised interval representation of $[0,1]$ and $\RR$.
		\item $f$ is a polytime computable point of the \parspace\ $C(I)_i$.
		\end{enumerate}
		\end{corollary}

		\Cref{resu: polytime computability of functions} in particular shows that
		for every polytime computable function $f \colon [0,1]_i \to \RR_i$ there exists
		a polytime computable $\psi \colon \ID \to \ID$ with
		$\psi \circ \varphi = f(\xi_{\RR_i}(\varphi))$ for all $\xi_{\RR_i}$-names $\varphi$.
		In other words, every polytime computable function is computed in polynomial
		time by an interval algorithm.

		Finally, consider the composition of functions, namely the operator
		\begin{equation*}\tag{comp}\label{eq:circ}
			\circ\colon C([0,1])\times C([0,1],[0,1]) \to C([0,1]),\quad (f,g)\mapsto f\circ g,
		\end{equation*}
		where $(f\circ g)(x):=f(g(x))$.
		Here $C([0,1],[0,1])$ is the set of all continuous functions on the unit interval whose image is contained in the unit interval.
		It makes sense to ask about the polytime computability of this operator whenever $C([0,1])$ is given the structure of a parametrised space, as it is possible to consider $C([0,1],[0,1])$ to be equipped with the range restriction of the representation of $C([0,1])$ and the restriction of the parameter to the domain of the range restriction.

		\begin{theorem}[Operations]\label{resu:operations}
			The arithmetic operations and composition are polytime computable on the interval functions.
		\end{theorem}

		\begin{proof}
			Witnesses of the polynomial-time computability of the arithmetic operations can easily be written down by doing interval arithmetic on the return values of the names of two functions.

			A realiser of the composition operator is the composition operator on Baire space.
			It remains to check, that the restriction of this operator to the domain of $\xi_{if}$ is a witness of the polytime computability of the composition of functions in the sense of \Cref{def:parameter polytime}.
			This can easily be checked by verifying that $\mu_{if}(\varphi\circ\psi)\leq \mu_{if}(\varphi)\circ\mu_{if}(\psi)$ for names of elements from the domain of the composition operator.
		\end{proof}

		The polytime computability of the arithmetic operations may also be deduced from
		the minimality of the interval-function representation and can in the process be generalised considerably.

		\begin{proposition}\label{resu: more operations}
				Let $H \colon \RR_i \times \RR_i \to \RR_i$ be a polytime computable function.
				Then the map
				\[C(I)_i \times C(I)_i \to C(I)_i, \; (f,g) \mapsto \left( x \mapsto H(f(x), g(x)) \right) \]
				is polytime computable.
		\end{proposition}
		\begin{proof}
			Consider the \parspace\ $\ZZ:= C(I)_i\times C(I)_i$ and the mapping
			\[ F\colon \ZZ \times I_i \to \RR_i, \; (f,g,x) \mapsto H(f(x),g(x)).\]
			As a composition of the polytime computable evaluation on $C(I)_i$ and $H$, the function $F$ is polytime computable.
			Use the currying property from \Cref{resu: currying} to obtain polytime computability of the function
			\[ c(F)\colon \ZZ \to C(I)_i,\quad c(F)(f,g)(x) = F(f,g,x). \]
			Note that $c(F)$ coincides with the function whose polytime computability was to be proven.
		\end{proof}
		In the above, $\RR_i$ could have been replaced by the Cauchy reals without
		changing the class of polytime
		computable functions on $\RR$ or $\RR\times \RR$.
		This class is the same as the one Ko introduces in his book \cite{MR1137517}.
		It should be noted that functions on unbounded domains are rarely considered complexity theoretically.
		In particular Ko seldom uses said definition in his book.

	\section{Comparison to Kawamura and Cook}\label{sec:comparison to represented spaces}

		The most used and best developed framework for complexity considerations in computable analysis is the framework of Kawamura and Cook.
		This framework is based on second-order complexity theory as presented in \Cref{sec:second-order complexity theory}.
		However, Kawamura and Cook add the following assumption about the elements that are allowed as names:
		\begin{definition}
			A string function $\varphi:\Sigma^*\to \Sigma^*$ is called \demph{length-monotone} if for all strings $\str a$ and $\str b$ we have
			\[ \length{\str a}\leq \length{\str b}\quad\Rightarrow\quad \length{\varphi(\str a)}\leq \length{\varphi(\str b)}. \]
			The set of all length-monotone string functions is denoted by $\reg$.
		\end{definition}
		Note that length-monotone string functions map strings of equal length to strings of equal length and we have $\length{\varphi}(n) = \length{\varphi(\sdzero^n)}$.
		\begin{definition}
			A representation is called \demph{length-monotone} if its domain is contained in $\reg$.
			A \parspace\ with a length-monotone representation and the standard parameter is called a \demph{Kawamura-Cook space}.
		\end{definition}
		The notion of polytime computability of functions between \KCspace s as \parspace s coincides with the definition given by Kawamura and Cook for represented spaces equipped with a length-monotone representation.

		The names in most representations considered in this paper are functions which return dyadic numbers.
		As any dyadic number has encodings of arbitrarily large size,
		we can pad an arbitrary name of an element to a length-monotone one.
		Thus, in many cases, representations can be restricted to $\reg$ to obtain a \KCspace\ from a represented space.
		Depending on the concrete example, this may or may not change the complexity structure.
		The restriction of the Cauchy representation of reals as introduced in \Cref{def:cauchy reals} to length-monotone names leads to a polytime equivalent representation.
		An example where the restriction to length-monotone names does make sense but is not polytime equivalent is the hyper-linear representation considered in the upcoming \Cref{sec:representations of continuous functions}.

		The main motivation for considering \KCspace s is that in these spaces the size of an oracle is accessible to a polytime machine.
		It is possible to prove that this characterizes the \KCspace s within the category of \parspace s up to isomorphism.
		Recall that this paper exclusively uses the unary encoding for natural numbers.
		Consider the following structure of a \parspace\ on the space of possible sizes of oracles:
		\begin{definition}
			Equip the space $\NN^\NN$ with the total representation defined by
			\[ \xi(\varphi)(n) = \length{\varphi(\sdone^n)}\]
			and the standard parameter.
		\end{definition}
		The space $(\NN^\NN,\xi,\mu)$ is a parametrised space:
		The identity can be computed by the machine that checks if the input is of the form $\sdone^n$, aborts if it is not and otherwise copies $\varphi$.
		\begin{lemma}
			Every second-order polynomial is polytime computable as a mapping from $\NN^\NN\times \NN$ to $\NN$.
		\end{lemma}
		The proof is a straightforward induction on the structure of second-order polynomials.
		This does not contradict the failure of time-con\-struct\-i\-bi\-li\-ty of second-order polynomials, as the difficult part, namely evaluating the size function of an oracle, has been skipped.

		Now the informal statement \lq the size of an oracle is accessible\rq\ can be replaced with the formal statement \lq the operation of sending an oracle to its size is polytime computable\rq\ and the promised characterisation of the \KCspace s follows.
		\begin{theorem}\label{resu: Characterisation of Kawamura-Cook space}
			A \parspace\ is polytime isomorphic to a \KCspace\
			if and only if it is polytime isomorphic to a space with polytime computable parameter.
		\end{theorem}

		\begin{proof}
			For the first direction assume that the \parspace\ $(X,\xi,\mu)$ has a polytime computable parameter.
			Let $N^?$ be a machine that computes the identity,
			running time and parameter blowup being bounded by a second-order polynomial $P$.
			Define a representation $\delta$ of $X$ as follows: A string-function $\varphi$ is a $\delta$-name of $x\in X$ if and only if there exists a $\xi$-name $\psi$ of $x$ such that
			\[ \varphi(\str a) = N^\psi(\str a)\sdzero\sdone^{P(\mu(\psi),\length{\str a}) - \length{N^\psi(\str a)}}. \]
			Note that since $P$ bounds the running time of $N^?$, the exponent above is always positive and we have $\length{\varphi(\str a)} = P(\mu(\psi),\length{\str a})+1$.
			In particular each such $\varphi$, and therefore also $\delta$, is length-monotone.
			Since the parameter is polytime computable and second-order polynomials can be evaluated in polynomial time, the obvious translation from $\xi$ to $\delta$ runs in polynomial time.
			A translation in the other direction can be computed by truncating the tail.
			The bounded parameter blowup of both of these translations follows from the equality $\length{\varphi}(n) = P(\mu(\psi),n) + 1$.

			For the other direction recall that the standard parameter is a restriction of the size function and due to the representation being length-monotone the restriction of the size function to the domain of the representation is polytime computable.
		\end{proof}

	\subsection[Representations of continuous functions]{Representations of $C([0,1])$}\label{sec:representations of continuous functions}

		Within their framework of second-order complexity theory, Kawamura and Cook have introduced
		a universal representation of univariate continuous functions
		based on the following classical characterisation of polytime computable real functions,
		see e.g.~\cite[Corollary 2.14]{MR1137517}:

		\begin{theorem}\label{resu: characterisation of polytime computability}
		A real function $f \colon [0,1] \to \RR$ is polytime computable if and only if the sequence
		$(f(d))_{d \in \DD \cap [0,1]}$ is a polytime computable sequence of reals and $f$ has a polynomial
		modulus of continuity.
		\end{theorem}
		Recall that a function $\nu:\NN\to\NN$ is called a \demph{modulus of continuity} of $f\in C([0,1])$ if we have
			\[ \forall x,y\in[0,1],\forall n\in\NN\colon \abs{x-y}\leq 2^{-\nu(n)}\Rightarrow \abs{f(x)-f(y)}\leq 2^{-n}. \]

		\begin{definition}[\cite{Kawamura:2012:CTO:2189778.2189780}]
			Define the \demph{\KCrep $\delta_{\square}$ of $C([0,1])$} as follows:
			A length-monotone function $\varphi:\DD\times \NN\to \DD$ is a name of $f\in C([0,1])$ if it satisfies
			\[ \forall r\in\DD, \forall n\in\NN\colon \abs{\varphi(r,n)-f(r)}\leq 2^{-n} \]
			and $\length\varphi$ is a modulus of continuity of $f$.
		\end{definition}
		Note that it is possible to require the names to be length-monotone and to use the size of a name for encoding the modulus since encodings of dyadic numbers can be chosen arbitrarily large.
		Let $C(I)_{KC}$ be the parametrised space induced by the representation $\delta_\square$ and the standard parameter.
		This space is obviously a \KCspace.
		Its representation is universal in the sense that	within the class of \KCspace s,
		it provides the minimal amount of information about a continuous function such that evaluation is possible in polynomial time.
		(For details see \cite{Kawamura:2012:CTO:2189778.2189780}, in particular Lemma~4.9.)
		Recall that \lq evaluation\rq\ refers to the functional
		\begin{equation*}
			\tag{eval}\label{eq:evaluation} \eval:C([0,1])\times [0,1]\to\RR,\quad (f,x)\mapsto f(x).
		\end{equation*}
		Kawamura and Cook's proof of minimality equips $[0,1]$ and $\RR$ with the length-monotone Cauchy representation and the standard parameter.
		Due to the polytime equivalence of the Cauchy and the interval reals from \Cref{resu:equivalence to the Cauchy representation}, this does not make a difference up to polytime equivalence.

		\begin{theorem}[Minimality \cite{Kawamura:2012:CTO:2189778.2189780}]\label{thm: Kawamura-Cook minimality}
			For a length-monotone representation $\delta$ the following are equivalent:
			\begin{itemize}
				\item Evaluation \eqref{eq:evaluation} is polytime computable.
				\item The representation $\delta$ is polytime translatable to $\delta_{\square}$.
			\end{itemize}
		\end{theorem}
		In particular, evaluation is polytime computable in the \KCrep.
		The above uses the concept of polytime computability used by Kawamura and Cook
		which is equivalent to having a polytime computable realiser.
		The notion is also equivalent to polytime computability if the spaces are turned
		into \KCspace s using the standard parameter.
		This means that the polytime translatability coincides with polytime reducibility as
		Kawamura and Cook consider it for length-monotone representations.

		\Cref{thm: Kawamura-Cook minimality} can also be stated in a weaker form that avoids mentioning length-monotonicity.
		Consider the multi-valued function which sends a real function to some modulus of continuity:
		\begin{equation*}\tag{mod}\label{eq:the modulus function}
			\mod:C([0,1]) \mto \NN^\NN,\quad f \mapsto \Set{\nu\in\NN^\NN}{\begin{matrix} \nu\text{ is mod. of cont. of } f \end{matrix}}.
		\end{equation*}
		We call this function the \demph{modulus-function} and it has to be multivalued to be computable:
		Since $C([0,1])$ is connected and $\NN^\NN$ is totally disconnected, the only single-valued continuous functions from $C([0,1])$ to $\NN^\NN$ are constant functions.
		\begin{theorem}[Minimality. General version]
			For a representation $\xi$ of $C([0,1])$ the following are equivalent:
			\begin{enumerate}
				\item Evaluation \eqref{eq:evaluation} and modulus \eqref{eq:the modulus function} are polytime computable.
				\item The representation $\xi$ is polytime translatable to $\delta_\square$.
			\end{enumerate}
		\end{theorem}

		A representation that can be regarded as an intermediate of the
		interval function representation and the \KCrep\ was introduced and investigated by Brau\ss e and Steinberg.
		\begin{definition}[\cite{Aminimal}]
			The \demph{hyper-linear representation $\xi_{hlin}$ of $C([0,1])$} is defined as follows:
			A function $\varphi:\DD\times \NN \to\DD\times \NN$, $\varphi(r,n)=(\varphi_1(r,n),\varphi_2(r,n))$ is a name of a function $f\in C([0,1])$ if
			\begin{itemize}
				\item $f([r\pm 2^{-\varphi_2(r,n)}])\subseteq [\varphi_1(r,n)\pm2^{-n}]$.
				\item $\forall r\in\DD,\forall n\in\NN\colon \varphi_2(r,n)\leq\length{\varphi}(n)$.
			\end{itemize}
		\end{definition}
		Denote the parametrised space that arises by equipping $C([0,1])$ with the hyper-linear representation and the standard parameter by $C(I)_{hlin}$.

		It is possible to restrict the hyper-linear representation to length-monotone names
		and the resulting length-monotone representation is polytime equivalent to the \KCrep $\delta_\square$.
		In \cite{Aminimal} it is proven that the hyper-linear representation itself is not polytime equivalent to any length-monotone representation.
		This is an example where the restriction to length-monotone names is possible but changes the complexity structure.
	\subsection{Comparison}\label{sec:comparison}
		Brau\ss{e} and Steinberg prove the following properties of the hyper-linear representation that are relevant for this paper:
		\begin{theorem}[\cite{Aminimal}]
			\begin{enumerate}
				\item[]
				\item The representation $\xi_{hlin}$ is polytime translatable to $\delta_\square$ but no polytime translation in the other direction exists.
				\item The modulus function is not polytime computable on $C(I)_{hlin}$.
				\item Composition is not polytime computable on $C(I)_{hlin}$.
			\end{enumerate}
		\end{theorem}
		They also provide a proof that the hyper-linear representation is
		the least representation such that hyper-linear-time evaluation is possible \cite{Aminimal}.
		The exact formulation of this result is rather technical and of no importance to the contents of this paper.
		It should be mentioned, however, that the restriction to hyper-linear-time is necessary and technical difficulties are due to
		hyper-linear-time computability not working	well as an replacement of polytime computability.
		Therefore the minimality property is slightly unsatisfactory and this, together with the failure of polytime
		computability of the composition, indicates that a different approach is necessary.

		The \KCrep is polytime translatable to the hyper-linear representation, which is polytime translatable to the interval function representation and none of these reductions reverse.
		In symbols, indicating the order by information content as discussed in \Cref{def: polytime translatability} this can be written as
		\[ C(I)_i<_P C(I)_{hlin}<_P C(I)_{KC}. \]
		Only the left inequality remains to be proven.
		\begin{figure}
		\centering
			\begin{centering}
				\begin{tabular}{c||c|c|c}
			 		&  $C(I)_{KC}$ & $C(I)_{hlin}$ & $ C(I)_i$ \\
			 		\hline\hline
			 		$\mod$ & polytime & \textcolor{red}{not p.t.} & \textcolor{red}{not p.t.} \\
			 		\hline
			 		$\eval,+,\times$ & polytime & polytime & polytime \\
			 		\hline\hline
			 		$\circ$ & polytime & \textcolor{red}{not p.t.} & polytime \\
			 		\hline
			 	\end{tabular}
			\end{centering}
		\caption{an overview over the complexity of operations with respect to different structures of $C([0,1])$ as represented or \parspace s.}\label{tabular}
		\end{figure}

		\begin{theorem}
			The hyper-linear representation can be translated to interval function representation in polynomial time.
			No polytime translation in the other direction exists.
		\end{theorem}

		\begin{proof}
			In \cite[Theorem 2.4]{Aminimal} it is proven that the hyper-linear representation renders evaluation polytime computable.
			By the minimality of the interval functions with respect to polytime evaluation from \Cref{resu:minimality} it follows, that the functions in the hyper-linear representation can be translated to the interval functions in polynomial time.

			On the other hand, composition is polytime computable on the interval functions, and is not polytime computable with respect to the hyper-linear representation \cite[Theorem 2.6]{Aminimal}.
			Thus, no translation in the other direction can exist.
		\end{proof}

		Since the \KCrep of continuous functions seems to be of more relevance than the hyper-linear representation we state the following corollary separately:

		\begin{corollary}\label{resu: interval function < Kawamura-Cook function}
			The interval function representation cannot be translated to the \KCrep in polynomial time.
			A polytime translation in the other direction does exists.
		\end{corollary}

		In \cite{Aminimal} it is proven that the modulus function is not polytime computable with respect to the hyper-linear representation.
		This has the following important implication for the \parspace\ of interval functions:

		\begin{corollary}\label{resu:modulus}
			The modulus function is not polytime computable on the \parspace\ of interval functions.
		\end{corollary}

		\begin{proof}
			In \cite[Theorem 2.6]{Aminimal} it is proven, that the modulus function is not polytime computable with respect to $\xi_{hlin}$.
			Since $\xi_{hlin}$ is polytime translatable to $(\xi_{if},\mu_{if})$, the same must hold for the \parspace\ of interval functions.
		\end{proof}

		Thus, while the polynomial computable points of $C(I)_i$ and $C(I)_{KC}$ are the same by
		\Cref{resu: polytime computability of functions},
		they do not behave equivalently in terms of uniform polytime computability,
		at least for multi-valued functions.
		Since the Kawamura-Cook space of continuous functions is universal amongst the Kawamura-Cook spaces
		which admit polytime evaluation, it follows that $C(I)_i$ is not isomorphic to any
		Kawamura-Cook space.
                
        An overview of the properties of the three spaces considered in this section can be found in \Cref{tabular}.
  		These properties show that the function space $C(I)_i$, unlike function spaces in previously
  		considered models, reflects the empirically observable properties of \irram-functions:
  		Evaluation and composition are easy to compute, while the modulus of continuity is hard to compute.
  		This is no coincidence: 
  		The Appendix discusses how the function space used in \irram can be modelled 
  		as a parametrised space which is polynomial-time isomorphic to $C(I)_i$.

\section{Conclusion}
		Kapron and Cook's characterisation of second-order polytime computability by means
		of resource bounded oracle machines came as quite a surprise.
		The basic feasible functionals satisfy the so-called Ritchie-Cobham property:
		the running time of an oracle machine which computes a basic feasible functional
		can be bounded by a basic feasible functional \cite{MR0411947}.
		However, while the size of return values of a basic feasible functional are always bounded by a second-order polynomial,
		for most second-order polynomials $P$ there is no basic feasible functional $T$ which satisfies
		\begin{equation}\label{eq: T > P}
		\forall \varphi\in\B,\str a\in \Sigma^* \colon \length{T(\varphi,\str a)} \geq P(\length{\varphi},\length{\str a}).
		\end{equation}
		This seems to suggest that the class of functions which can be computed
		with a second-order polynomial time bound
		is strictly larger than the class of basic feasible functionals.
		The reason for the characterisation to still go through is that the oracle access
		of oracle machines is very restricted already by second-order polynomial time bounds, prohibiting detection of big inputs.

		Using the characterisation as a definition, the proof of the Ritchie-Cobham
		property makes essential use of the totality of the functionals at hand.
		While	totality is a standard assumption in second-order complexity theory,
		applications in computable analysis require a very specific kind of partiality.
		From this point of view, the failure of \Cref{eq: T > P} may be considered pathological.
		Adding length-monotonicity of names as an assumption removes this pathological behaviour,
		as the restriction of the size function to the length-monotone functions is polytime computable.
		Furthermore, it seems to be a necessary restriction for the proof that Kawamura and Cook provide
		that their representation of the continuous functions on the unit interval has the minimality property.

		We believe that the framework of Kawamura and Cook is the best solution that allows for formulation within second-order complexity theory in its traditional form.
		The extension proposed in this paper does not allow formulation within the scope of second-order complexity in the sense that whether an algorithm runs in polynomial time is allowed to depend on properties of its inputs,
		namely the parameter, that need not make sense independently of the interpretation of the inputs.
		This is the reason for the use of the term \lq parameter\rq\ that indicates that
		semantic information about the objects has to be taken into account in addition to the size of a name.
		We believe this step to be necessary to allow for the description of the behaviour of
		efficient software based on computable analysis.
		Previous attempts to recover a description within second-order complexity theory have run into serious problems
		that where due to specific properties of the size function and the necessity to modify it on the output side \cite{Aminimal}.

		It is not very difficult to produce \parspace s that are likely to be of importance
		and not isomorphic to any space that uses the standard parameter.
		Sierpinski space with the standard representation and the parameter that assigns the position of the first
		non-zero value to a name of $\top$ and zero to the unique name of $\bot$ is an example of such a space.
		We conjecture that the \parspace\ of interval functions is another example.
		A reason for this being a tricky question could be that the running time of
		the parameter of the interval functions and the
		running time of the size function on Baire spaces are fairly similar.
		An affirmative answer to this question would provide an even stronger
		motivation for the introduction of \parspace s.

		While we started the investigation of \parspace s out of what we believe to be necessity,
		we have found them to	provide a very natural framework for complexity considerations in computable analysis.
		We believe that \parspace s broaden the scope of real complexity theory considerably and
		we intend to find new applications and explore the limits of the method.
		We expect that additional insight can already be gained by re-evaluating existing work in this generalised setting.
		The rest of this conclusion lists some ideas we consider especially promising and plan to pursue in the future.

		A special case of the following is discussed in
		\Cref{resu: polytime function on interval reals has first-order time bound}:
		All of \cite{MR3219039}, \cite{MR2090390} and \cite{SteinbergPhD}
		use very similar constructions to produce certain kinds of counterexamples.
		They give general constructions which make an arbitrary given representation
		into one with highly pathological properties.
		In the category of \parspace s these constructions
		can be modified so as to remove their pathological behaviour:
		If, in addition to modifying the representation,
		one modifies the parameter appropriately,
		one obtains a polytime isomorphic \parspace.
		This space has the desirable property that polytime computability
		can be characterised using first-order time bounds in the running time constraint.
		This may be desirable as it allows to separate an algorithm into two parts:
		Firstly, a complexity-theoretical part that is mostly about efficiently
		manipulating approximations and may stay in a first-order setting.
		Secondly, a mathematical part that is about providing bounds on rates of convergence
		where the second-order nature of the problem at hand can no longer be ignored.
		A detailed description of this will be part of future work.
		For now we have to point to
		\Cref{resu: polytime function on interval reals has first-order time bound}
		for special case of this.

		The present work proves that the interval reals are polytime equivalent to the Cauchy reals.
		However, we suspect that it is no coincidence that interval representations
		are favoured over the Cauchy representation in practice.
		In practice, memory consumption is often significantly more important than
		running time.
		Libraries like \irram are developed with this in mind and value space efficiency over time efficiency.
		It is therefore desirable to settle the question whether the interval reals
		and the Cauchy reals are also logarithmic space equivalent.
		Unfortunately, the tools to investigate this question are not available to us at this point in time.
		There is work about space restricted computation in the framework of Kawamura and Cook that one can check results about
		\parspace s against \cite{MR3259646} and we are highly interested in pursuing this line of research in the future.

		Recently there has been some interest in the study of exponentiable
		objects in the category of represented spaces with polytime computable
		functions as morphisms \cite{MR3219039, higherorder}.
		We hope that the study of parametrised spaces can shed some light on this,
		as it seems like parametrised spaces allow for more natural exponential
		constructions than \KCspace s.
		For instance, our parametrised space of interval functions (\Cref{def: interval function representation})
		enriches a natural exponential representation with a suitable parameter,
		whilst the Kawamura-Cook representation of $C(I)$ requires the hard-coding
		of a modulus of uniform continuity.
		It would be interesting to investigate to which extent this construction can
		be generalised.
		Exponentials in the category of parametrised spaces should yield
		exponentials in the category of \KCspace s by enriching the representation
		with bounds on the parameter.
	\bibliography{bib}{}
\newpage
\begin{appendix}
\section{Non-monotone interval enclosures}\label{sec: non-monotone enclosures}

From a theoretical point of view, the equivalence to the Cauchy representation from \Cref{resu:equivalence to the Cauchy representation}
justifies our choice of a representation and parameter for the real numbers and the minimality \Cref{resu:minimality}
supports that the choices for the real functions are satisfactory.
Still, to describe the representations used in certain implementations precisely,
one might have to relax the definitions somewhat.
Notably, the \irram library works with enclosures which are not necessarily monotone.
We show that this choice is inconsequential for real numbers, in the sense that
the resulting parametrised space is equivalent to the space of interval reals.
It does however cause difficulties in the treatment of real functions,
as the natural function space construction on ``\irram reals''
does not allow for a well-defined extension of the
function space parameter on interval functions.
On the other hand, the restriction of this representation to names with well-defined parameters
is again polytime equivalent to the interval function representation.
Roughly speaking, a name has a well-defined parameter if it encodes a modulus of continuity
of the function it is representing.
We argue that all ``reasonable'' algorithms define names with well-defined parameters.

\begin{definition}\label{def: iRRAM reals}
A function $\varphi \colon \NN \to \ID$ is a $\xi_{\RR_{\irram}}$-name of $x \in \RR$ if
\begin{enumerate}
\item $x \in \varphi(n)$ for all $n \in \NN$.
\item $\varphi(n) \to \{x\}$ in the Hausdorff metric as $n \to \infty$.
\end{enumerate}
Define the space of \irram reals to be the parametrised space
\[\RR_{\irram} = (\RR, \xi_{\RR_{\irram}}, \mu_{\irram}),\]
where $\mu_{\RR_{\irram}}$ is the natural extension of $\mu_{\RR_i}$ to the domain of $\xi_{\RR_{\irram}}$:
\[\mu_{\RR_{\irram}}(\varphi)(n) = \min\{N \in \NN \mid \forall m \geq N: \diam(\varphi(m)) \leq 2^{-m}\}
			+  \lceil\lb(\abs{\xi_{\RR_{\irram}}(\varphi)} + 1)\rceil.  \]
\end{definition}

The proof that $\RR_{\irram}$ is indeed a parametrised space proceeds similar to the same proof for the interval reals from \Cref{resu:the parametrised space of irram reals}.
Note that unlike the parameter $\mu_{\RR_i}$, the parameter $\mu_{\RR_{\irram}}$ is not computable.
Nevertheless, we have the following result:

\begin{proposition}[$\RR_{\irram}\equiv_P\RR_i$]\label{resu: interval reals = irram reals}
	The space of \irram reals and the space of interval reals are isomorphic as \parspace s.
\end{proposition}
\begin{proof}
Clearly, $\RR_i$ translates into $\RR_{\irram}$ (using \Cref{resu:the parametrised space of irram reals}),
so we just have to find a translation in the opposite direction.
Suppose we are given a $\xi_{\RR_{\irram}}$-name $\varphi$ of some $x \in \RR$.
Again, up to applying a polytime computable realiser of the identity
as in \Cref{resu:the parametrised space of irram reals},
we can assume that the $\varphi$ has polynomially bounded size.
Then we can compute in polynomial time the function
$\psi(n) = \bigcap_{k = 0}^n \varphi(k)$
which is a $\xi_{\RR_i}$-name of $x$ with
$\mu_{\RR_i}(\psi(n)) = \mu_{\RR_{\irram}}(\varphi(n))$.
\end{proof}

The canonical function space construction from Definition \ref{def: interval function representation}
immediately yields an analogous representation for the functions on \irram reals.

\begin{definition}
Define the \demph{\irram function representation $\xi_{\irram f}$} as follows: A function
$\psi \colon \ID \to \ID$ is a name of $f \in C([0,1])$ if and only if
\[\forall \varphi \in \dom (\xi_{\RR_\irram}):
	 \left(\xi_{\RR_{\irram}}(\varphi) \in [0,1] \Rightarrow \xi_{\RR_{\irram}}(\psi \circ \varphi) = f(\xi_{\RR_{\irram}}(\varphi)) \right). \]
\end{definition}

Note that the domain of $\xi_{\irram f}$ is strictly larger than the domain of $\xi_{if}$.
In fact, the $\xi_{if}$-names are precisely the monotone $\xi_{\irram f}$-names.
The representation $\xi_{\irram f}$ has the disadvantage that the parameter $\mu_{if}$ cannot be extended to the whole domain of the representation:

\begin{example}
Consider the function $\psi\colon\ID\to\ID$ defined by
	\[
	\psi(\str a) :=
	\begin{cases}
					[\frac12\pm\frac12]& \text{if } \str a= [3\cdot 2^{-n-2}\pm2^{-n-2}]\text{ for some $n$}\\
			[0\pm 0] & \text{otherwise.}
	\end{cases}
	\]
This function is a $\xi_{\irram f}$-name of the zero function and the minimum in the definition of $\mu_{if}$ is undefined.
\end{example}

Therefore in order to obtain a parametrised space one has to restrict the representation further:

\begin{definition}
The \demph{\parspace{} of \irram functions} is the \parspace{}
$C(I)_{\irram} = (C([0,1]), \xi_{\irram f}|_{\dom(\mu_{if})}, \mu_{if})$.
\end{definition}
By \Cref{resu: interval function parameter well-defined}, monotonicity is a sufficient condition for the parameter to be well-defined.
Even more generally, by essentially the same argument,
continuity with respect to the Hausdorff metric on the dyadic
intervals is sufficient.
Hence any ``reasonable'' interval algorithm computes a function with well-defined parameter.

\begin{proposition}[$C(I)_{\irram}\equiv_PC(I)_i$]\label{resu: irram funs and interval funs}
The space of \irram functions and the space of interval functions are isomorphic as \parspace s.
\end{proposition}
\begin{proof}
It is easy to see that $\xi_{if}$ reduces in polynomial time to $\xi_{\irram f}$, using some polytime realiser of the
identity in $C(I)_i$ together with the fact that $\xi_{if}$-names are precisely the monotone $\xi_{\irram f}$-names.
Conversely, on the space $C(I)_{\irram}$ the evaluation functional
$\eval \colon C(I)_{\irram}\times [0,1]_{i} \to \RR_{i}$
is easily seen to be polytime computable (using \Cref{resu: interval reals = irram reals}).
It follows from \Cref{resu:minimality} that $\xi_{\irram f}$ reduces in polynomial time to $\xi_{if}$.
\end{proof}
\end{appendix}

\end{document}